\documentclass[11pt,reqno]{amsart}

\usepackage{amsmath, amssymb, amsthm, amsfonts, cite, enumitem, color}
\usepackage{a4wide}
\usepackage{amssymb}
\usepackage{pdfsync}
\usepackage{hyperref}
\usepackage{graphics,pstricks}

\usepackage{tikz}

\usepackage[margin=25mm]{geometry}

\newtheorem{theorem}{Theorem}[section]
\newtheorem{proposition}[theorem]{Proposition}

\newtheorem{definition}[theorem]{Definition}
\newtheorem{lemma}[theorem]{Lemma}

\numberwithin{equation}{section}

\newcommand{\opH}{H_{\omega,\lambda}}
\newcommand{\T}{{\bf T}_{\omega,\lambda}}

\newcommand{\vp}{\phi}% eigenfunction

\newcommand{\Z}{\mathbb{Z}}
\newcommand{\R}{\mathbb{R}}

\newcommand{\p}{\mathbb{P}}
\newcommand{\esp}{\mathbb{E}}

\newcommand{\bra}{\langle}
\newcommand{\ket}{\rangle}

\newcommand{\norm}[1]{\left\lVert #1 \right\rVert}

\newcommand{\e}{\mathrm{e}}

\def \simless {\mathbin{\lower 3pt\hbox{$\rlap{\raise 5pt
              \hbox{$\char'074$}}\mathchar"7218$}}}

\author[O. Bourget, G. R. Moreno Flores and A. Taarabt]{Olivier Bourget$^1$, Gregorio R. Moreno Flores$^{2,*}$ and Amal Taarabt$^3$}

\thanks{{\it Key words and phrases.} 
Anderson model, decaying disordered, dynamical localization.}

\thanks{ AMS 2010 {\it subject classifications}. 82B44, 47B80}

\thanks{$^*$ Corresponding author}

\thanks{$^1$  Partially supported by Fondecyt grant 1161732}

\thanks{$^1$ $^2$ $^3$ Facultad de Matem\'aticas, Pontificia Universidad Cat\'olica de Chile.}

\thanks{$^2$  Partially supported by Fondecyt grant 1171257, N\'ucleo Milenio `Modelos Estoc\'asticos de Sistemas Complejos y Desordenados' and MATH Amsud `Random Structures and Processes in Statistical Mechanics'}

\thanks{$^3$  Partially supported by Fondecyt grant 11190084}

\address{Facultad de Matem\'aticas\\
Pontificia Universidad Cat\'olica de Chile\\
Vicu\~na Mackenna 4860, Macul\\
Santiago, Chile}

\email{bourget@mat.uc.cl, grmoreno@mat.uc.cl, amtaarabt@mat.uc.cl}

\title[Anderson Model in a Decaying Potential]{Dynamical Localization for the One-dimensional Continuum Anderson Model in a Decaying Random Potential}

\begin{document}

\begin{abstract}
	%We show dynamical localization for the one-dimensional continuum Anderson model in a sub-critical decaying random potential.
	
	We consider a one-dimensional continuum Anderson model where the potential decays in average like $|x|^{-\alpha}$, $\alpha>0$. We show dynamical localization for $0<\alpha<\frac12$ and provide  control on the decay of the eigenfunctions.
	
%	This simple model is known to display a rich phase diagram with different kinds of spectrum arising as the decay rate $\alpha$ varies. 
%	In particular, the spectrum of the operator is almost surely pure point for $0<\alpha<1/2$.

%	In \cite{KLS}, the authors show a.c. spectrum in the super-critical case $\alpha>1/2$, a transition from singular continuous to pure point spectrum in the critical case $\alpha=1/2$ and point spectrum for $\alpha<1/2$.
	%At the dynamical level, delocalization holds for $\alpha>1/2$. The work \cite{GKT} shows that there is no dynamical transition for $\alpha=1/2$ despite of the spectral transition by exhibiting non-trivial transport exponents for all energies. 
\end{abstract}

\maketitle

\tableofcontents

%%%%%%%%%%%%%%%%%%%%%%%%%%%%%%%%%%%%%%%%%%%%%%%%%%%%%%%%%%%%
%%%%%%%%%%%%%%%%%%%%%%%%%%%%%%%%%%%%%%%%%%%%%%%%%%%%%%%%%%%%
%%%%%%%%%%%%%%%%%%%%%%%%%%%%%%%%%%%%%%%%%%%%%%%%%%%%%%%%%%%%
%%%%%%%%%%%%%%%%%%%%%%%%%%%%%%%%%%%%%%%%%%%%%%%%%%%%%%%%%%%%
%%%%%%%%%%%%%%%%%%%%%%%%%%%%%%%%%%%%%%%%%%%%%%%%%%%%%%%%%%%%
%%%%%%%%%%%%%%%%%%%%%%%%%%%%%%%%%%%%%%%%%%%%%%%%%%%%%%%%%%%%
%%%%%%%%%%%%%%%%%%%%%%%%%%%%%%%%%%%%%%%%%%%%%%%%%%%%%%%%%%%%

%%%%%%%%%%%%%%%%%%%%%%%%%%%%%%%%%%%%%%%%%%%%%%%%%%%%%%%%%%%%
%%%%%%%%%%%%%%%%%%%%%%%%%%%%%%%%%%%%%%%%%%%%%%%%%%%%%%%%%%%%
%%%%%%%%%%%%%%%%%%%%%%%%%%%%%%%%%%%%%%%%%%%%%%%%%%%%%%%%%%%%
%%%%%%%%%%%%%%%%%%%%%%%%%%%%%%%%%%%%%%%%%%%%%%%%%%%%%%%%%%%%
%%%%%%%%%%%%%%%%%%%%%%%%%%%%%%%%%%%%%%%%%%%%%%%%%%%%%%%%%%%%
%%%%%%%%%%%%%%%%%%%%%%%%%%%%%%%%%%%%%%%%%%%%%%%%%%%%%%%%%%%%
%%%%%%%%%%%%%%%%%%%%%%%%%%%%%%%%%%%%%%%%%%%%%%%%%%%%%%%%%%%%

\section{Introduction}

Disordered systems in material sciences have been the source of a plethora of interesting phenomena and many practical applications. The addition of impurities in otherwise fairly homogeneous materials is known to induce new behaviours such as Anderson localization where wave packets get trapped by the disorder and conductivity can be suppressed \cite{An58}. It is then natural to expect that accurate mathematical models for disordered media should display an interesting phase diagram. 

As a model for the dynamics of an electron in a disordered medium, the Anderson model is expected to undergo a transition from a delocalized to a localized regime reflected at the spectral level by a transition from absolutely continuous to pure point spectrum. While the localized regime is well understood (see \cite{AW, CBD} and references therein), the existence of absolutely continuous spectrum remains a mystery (nonetheless, see \cite{Kl,ASW,FHS,GKS}).  

In order to understand how absolutely continuous spectrum survives in spite of the disorder, it has been proposed to modulate the random potential by a decaying envelope \cite{DSS, D, FGKM,KKO,Kr,KU}, this is, to replace the usual i.i.d. random variables $\{V(n):\, n\in\Z^d\}$ by $a_n V(n)$, where $(a_n)_n$ is a deterministic sequence satisfying $a_n \sim |n|^{-\alpha}$ for some decay rate $\alpha>0$.
%This model is well studied in dimension one \cite{Si95,KRS} and a critical value of $\alpha$ gives a rise to a transition from point spectrum to continuous and singular spectrum \cite{D}.
For large values of $\alpha$ and dimensions $d\geq 3$, scattering methods can be applied, leading to the proof of absolutely continuous spectrum \cite{Kr}. A wider range of values of $\alpha$ was considered by Bourgain in dimension $2$ \cite{B1} and higher \cite{B2}. Point spectrum was also showed to hold outside the essential spectrum of the operator in \cite{KKO}. 

It is well known that, in the i.i.d. case, the one-dimensional Anderson model always displays pure point spectrum
\cite{Car,GMP,KS,DSS,BDFG,DG,GZ,JZ,HSS,DaSiSt} while the addition of a decaying envelope leads to a rich phase diagram as the value of $\alpha$ varies. Transfer matrix analysis can be applied, leading to a complete understanding of the spectrum of the operator \cite{DSS, KLS} in the discrete and continuum setting (see also \cite{KU} for a related model). This time, absolutely continuous spectrum can still be observed for large values of $\alpha$. As it is natural to expect, small values of $\alpha$ lead to pure point spectrum. Interestingly, there is a critical value of $\alpha$ for which a transition from pure point to singular continuous spectrum occurs as a function of the coupling constant. The three above regimes correspond to $\alpha>\frac12$, $\alpha<\frac12$ and $\alpha=\frac12$ respectively. A complete study of the spectral behaviour of the one-dimensional discrete and continuum models is given in \cite{KLS} (see also \cite{BMT00}).

From the dynamical point of view, it is standard to show that the system propagates for $\alpha>\frac12$. For the critical case $\alpha=\frac12$, no transition occurs at the dynamical level, despite of the spectral transition: there are non-trivial transport exponents for all values of the coupling constant \cite{GKT} for both the discrete and continuum model (see also \cite{BMT00,BMT01} for elementary arguments showing delocalization). This provides yet another example of a model where spectral localization and transport coexist. Dynamical localization in the regime $0<\alpha<\frac12$ for the discrete model was shown in \cite{Si82}. 

In the present paper, we show dynamical localization for the continuum model in the sub-critical region $0<\alpha<\frac12$. This was left as an open question in \cite{DS}, where the authors develop a continuum version of the Kunz-Souillard method \cite{KS}. Instead, we have chosen to work within the framework of the continuum fractional moment method of \cite{AENSS} in the one dimensional version of \cite{HSS}. In addition, our proof involves a fine tuning of various auxiliary technical ingredients scattered in the literature e.g. \cite{KLS, CKM, CHN}.

\subsection*{Structure of the article}

We present the model and the main results in Section \ref{sec:model-and-results}. In Section \ref{sec:transfer-matrices} we recall the Pr\"ufer transform formalism and give some preliminary bounds. The fractional moments of the Green's function are studied in Section \ref{sec:FMM}. In Section \ref{sec:DL}, we prove the main theorem on dynamical localization. Section \ref{sec:consequences} contains the proof of some consequences of our main result. Finally, the Appendix contains some estimates used along the proofs.

%%%%%%%%%%%%%%%%%%%%%%%%%%%%%%%%%%%%%%%%%%%%%%%%%%%%%%%%%%%%
%%%%%%%%%%%%%%%%%%%%%%%%%%%%%%%%%%%%%%%%%%%%%%%%%%%%%%%%%%%%
%%%%%%%%%%%%%%%%%%%%%%%%%%%%%%%%%%%%%%%%%%%%%%%%%%%%%%%%%%%%
%%%%%%%%%%%%%%%%%%%%%%%%%%%%%%%%%%%%%%%%%%%%%%%%%%%%%%%%%%%%

\section{Model and main results}\label{sec:model-and-results}

%%%%%%%%%%%%%%%%%%%%%%%%%%%%%%%%%%%%%%%%%%%%%%%%%%%%%%%%%%%%
%%%%%%%%%%%%%%%%%%%%%%%%%%%%%%%%%%%%%%%%%%%%%%%%%%%%%%%%%%%%
%%%%%%%%%%%%%%%%%%%%%%%%%%%%%%%%%%%%%%%%%%%%%%%%%%%%%%%%%%%%
%%%%%%%%%%%%%%%%%%%%%%%%%%%%%%%%%%%%%%%%%%%%%%%%%%%%%%%%%%%%
%%%%%%%%%%%%%%%%%%%%%%%%%%%%%%%%%%%%%%%%%%%%%%%%%%%%%%%%%%%%
%%%%%%%%%%%%%%%%%%%%%%%%%%%%%%%%%%%%%%%%%%%%%%%%%%%%%%%%%%%%
%%%%%%%%%%%%%%%%%%%%%%%%%%%%%%%%%%%%%%%%%%%%%%%%%%%%%%%%%%%%
Let $\{\omega_n:\, n\in\Z\}$ be an i.i.d. family of bounded random variables defined on a probability space $(\Omega,\mathcal{F},\p)$. Assume that $\esp[\omega_0]=0$ and $\esp[\omega_0^2]=1$. We denote by $\omega^+$ and $\omega^-$ the supremum and infimum of the support of the $\omega_n$'s. By choosing an appropriate probability space, we may assume that these random variables are bounded and not merely bounded with probability $1$. We will always assume that the random variables $\omega_n$ admit a bounded density $\rho$ which is hence supported on $[\omega^-,\omega^+]$. 
%Note that the probability space can be constructed in such a way that the random variables are bounded and not just almost-surely bounded. 
\newline
Let $u$ be a non-negative bounded function with support in $(0,1)$ and denote $u_n(\cdot)=u(\cdot-n)$. This is often called the single-site potential. We assume that there exists a non-trivial interval $J\subset (0,1)$ and two constants $c_u,C_u\in (0,\infty)$ such that
\begin{eqnarray}\label{eq:single-site-potential}
	c_u \chi_J \leq u \leq C_u \chi_{[0,1]}.
\end{eqnarray}
Finally, let $\lambda\ne0$ and $\alpha>0$, and let $(a_n)_n$ be a positive sequence such that $\displaystyle\lim_{|n|\to\infty} a_n |n|^{\alpha}=1$.

We consider the random operator
\begin{eqnarray}
	\opH= - \Delta + \lambda V_{\omega} \quad \mathrm{on}\quad L^2(\R),
\end{eqnarray}
where $V_\omega$ is a multiplication operator by the function
 \begin{eqnarray}
 	V_{\omega}(x) = \sum_{n\in\Z} a_n \omega_n u_n(x).
 \end{eqnarray}
Notice that $(\opH)_{\omega}$ is a non-ergodic family of random operators which are essentially self-adjoint on $C_0^{\infty}(\R)$, the space of infinitely differentiable compactly supported functions. 
%%%%%%%%%%%%%%%%%%%%%%%%%%%%%%%%%%%%%%%%%%%%%%%%%%%%%%%%%%%%
%%%%%%%%%%%%%%%%%%%%%%%%%%%%%%%%%%%%%%%%%%%%%%%%%%%%%%%%%%%%
%%%%%%%%%%%%%%%%%%%%%%%%%%%%%%%%%%%%%%%%%%%%%%%%%%%%%%%%%%%%
%%%%%%%%%%%%%%%%%%%%%%%%%%%%%%%%%%%%%%%%%%%%%%%%%%%%%%%%%%%%

We start recalling the spectral results of \cite{KLS} which give a complete characterization of the spectrum of $\opH$ for all the possible combinations of parameters. The following is proved for the model on the half-line. 
%%%%%%%%%%%%%%%%%%%%%%%%%%%%%%%%%%%%%%%%%%%%%%%%%%%%%%%%%%%%
%%%%%%%%%%%%%%%%%%%%%%%%%%%%%%%%%%%%%%%%%%%%%%%%%%%%%%%%%%%%
\begin{theorem}\label{thm:KLS}
 Under the hypothesis above, the essential spectrum of $\opH$ is $\p$-a.s. equal to $[0,\infty)$. Furthermore, 

\begin{enumerate}
	\item[(1)] \textbf{Super-critical case.} If $\alpha>\frac12$ then for all $\lambda\in\mathbb{R}$, the spectrum of $\opH$ is almost surely purely absolutely continuous in $(0,\infty)$.
					\vspace{1ex}

    \item[(2)] \textbf{Critical case.} If $\alpha=\frac12$ then for all $\lambda \neq 0$, the a.c. spectrum of $\opH$ is almost surely empty. Furtheremore, for each $\lambda \neq 0$, there exists $E_0(\lambda)\geq 0$ such that, almost surely, the spectrum of $\opH$ is pure point in $(0,E_0(\lambda))$ and purely singular continuous in $[E_0(\lambda),\infty)$.

	\vspace{1ex}
		
	\item[(3)] \textbf{Sub-critical case.} If $0<\alpha<\frac12$ then for all $\lambda\neq 0$, the spectrum of $\opH$ is almost surely pure point in $(0,\infty)$.
\end{enumerate}
\end{theorem}
%%%%%%%%%%%%%%%%%%%%%%%%%%%%%%%%%%%%%%%%%%%%%%%%%%%%%%%%%%%%
%%%%%%%%%%%%%%%%%%%%%%%%%%%%%%%%%%%%%%%%%%%%%%%%%%%%%%%%%%%%
Delocalization, or spreading of wave packets, for $\alpha>\frac12$ follows from the RAGE theorem \cite{CFKS}. The situation is particularly interesting for $\alpha=\frac12$: non-trivial transport occurs regardless of the precise nature of the spectrum. In particular, this provides an example of an operator displaying pure point spectrum but no dynamical localization. To describe the dynamics, we consider the time-averaged random moment of order $p\ge0$ for the evolution, 
%in the operator norm, 
initially spatially localized at the origin and localized in energy by a positive function $f\in C_{0}^\infty(\R)$,
\begin{equation*}
	\mathbb{M}_{\omega}(p,f,T):=\frac{2}{T}\int^{\infty}_0 \e^{-\frac{2t}{T}} \left\| |X|^{\frac{p}{2}} \e^{-it\opH}f(\opH) \chi_0 \right\|^2 dt,
\end{equation*}
where $|X|$ denotes the position operator and $\chi_x$ denotes the characteristic function of the interval $[x,x+1]$.
The following result is proved in \cite{GKT} for the model defined on the half-line. 
%%%%%%%%%%%%%%%%%%%%%%%%%%%%%%%%%%%%%%%%%%%%%%%%%%%%%%%%%%%%
%%%%%%%%%%%%%%%%%%%%%%%%%%%%%%%%%%%%%%%%%%%%%%%%%%%%%%%%%%%%
\begin{theorem}\label{thm:GKT}
	Let $\alpha=\frac12$ and $\lambda \in \R$. The following holds $\p$-amost surely: for all positive $f\in C^{\infty}_0(\R)$ constantly equal to $1$ on a compact interval $J\subset (0,\infty)$, for any $\nu>0$ and all $p>2\gamma_J +\nu$ where 
	$\gamma_J=\inf\{\lambda(8E)^{-1}|\widehat{u}(\sqrt{E})|:\, E\in J\}$, 
	there exists $C_{\omega}(p,J,\nu)>0$ such that
	\begin{eqnarray*}
		\mathbb{M}_{\omega}(p,f,T) \geq C_{\omega}(p,J,\nu) T^{p-2\gamma_J-\nu},
	\end{eqnarray*}
	for all sufficiently large values of $T$ and where $\widehat u$ denotes the Fourier coefficient of $u$.
\end{theorem}
%%%%%%%%%%%%%%%%%%%%%%%%%%%%%%%%%%%%%%%%%%%%%%%%%%%%%%%%%%%%
%%%%%%%%%%%%%%%%%%%%%%%%%%%%%%%%%%%%%%%%%%%%%%%%%%%%%%%%%%%%
We complete the dynamical study of the model by addressing the question of dynamical localization for $0<\alpha<\frac12$. The corresponding result for the discrete model was obtained in \cite{Si82}.

 For an interval of energy $I$, we define the correlator 
\begin{eqnarray}\label{eq:correlator}
	Q_{\omega,\lambda}(x,y;I)
	=\sup_{\substack{f\in \mathcal{C}_{c}(I)\\ \norm{f}\le1}}
	\norm{ \chi_x f(\opH) \chi_y},
\end{eqnarray}
where $\mathcal{C}_c(I)$ denotes the space of bounded measurable functions compactly supported in $I$.
%%%%%%%%%%%%%%%%%%%%%%%%%%%%%%%%%%%%%%%%%%%%%%%%%%%%%%%%%%%%
%%%%%%%%%%%%%%%%%%%%%%%%%%%%%%%%%%%%%%%%%%%%%%%%%%%%%%%%%%%%
We now give the notion of dynamical localization we will use in this work.
%%%%%%%%%%%%%%%%%%%%%%%%%%%%%%%%%%%%%%%%%%%%%%%%%%%%%%%%%%%%
%%%%%%%%%%%%%%%%%%%%%%%%%%%%%%%%%%%%%%%%%%%%%%%%%%%%%%%%%%%%
\begin{definition}
We say that $\opH$ exhibits \textit{dynamical localization} in an interval $I\subset\R$ if we have
\begin{equation}\label{eq:dynamical-localization}
\sum_{x\in\Z} \esp\left[ Q_{\omega,\lambda}(x,y;I)^2\right]<\infty,
\end{equation}
for all $y\in\Z$.
\end{definition}
%%%%%%%%%%%%%%%%%%%%%%%%%%%%%%%%%%%%%%%%%%%%%%%%%%%%%%%%%%%%
%%%%%%%%%%%%%%%%%%%%%%%%%%%%%%%%%%%%%%%%%%%%%%%%%%%%%%%%%%%%
Our main theorem is the following.
\begin{theorem}\label{thm:DL}
	Let $0<\alpha<\frac12, \lambda\neq 0$ and let $I\in(0,\infty)$ be a compact interval.
	For each $y\in\Z$, there exist constants $c=c(I,y),\, C=C(I,y) \in (0,\infty)$ such that
	\begin{eqnarray*}
		\esp\left[ Q_{\omega,\lambda}(x,y;I) \right] 
		\leq
		C
		\e^{-c|x|^{1-\alpha}},
	\end{eqnarray*}
	for all $x\in\R$. In particular, dynamical localization in the sense of \eqref{eq:dynamical-localization} holds in any compact subinterval of $(0,\infty)$.
\end{theorem}
%%%%%%%%%%%%%%%%%%%%%%%%%%%%%%%%%%%%%%%%%%%%%%%%%%%%%%%%%%%%
%%%%%%%%%%%%%%%%%%%%%%%%%%%%%%%%%%%%%%%%%%%%%%%%%%%%%%%%%%%%
Although the lack of ergodicity of the model induces the dependence of \eqref{eq:dynamical-localization}  on the base site $y$, it is standard to show that this bound still implies pure point spectrum and finiteness of the moments. We will recall the proof of pure point spectrum in Appendix \ref{app:pp} as some care is needed to overcome the non-uniform bounds. Respect to the finiteness of the moments, we will prove a stronger result in Theorem \ref{thm:lower-bound-DL}.

%%%%%%%%%%%%%%%%%%%%%%%%%%%%%%%%%%%%%%%%%%%%%%%%%%%%%%%%%%%%
%%%%%%%%%%%%%%%%%%%%%%%%%%%%%%%%%%%%%%%%%%%%%%%%%%%%%%%%%%%%
%\begin{proposition}\label{thm:consequences-DL}
%	Assume dynamical localization for $\opH$ holds in the sense of \eqref{eq:dynamical-localization} in an energy interval $I\subset\R$. Then the following holds:
%	\begin{enumerate}
%		\item[(i)] The spectrum of $\opH$ is almost surely pure point in $I$.
%		
%		\item[(ii)] For all $p>0$,
%			\begin{equation*}
%				\esp\left(\sup_{t\in\R}
%					\norm{
%						|X|^p\e^{-it\opH}P_I(\opH)\varphi_0}^2
%				\right)<\infty,
%			\end{equation*}
%			for all $\varphi_0 \in L^2(\R)$ with bounded support. 
%	\end{enumerate}
%\end{proposition}
%%%%%%%%%%%%%%%%%%%%%%%%%%%%%%%%%%%%%%%%%%%%%%%%%%%%%%%%%%%%
%%%%%%%%%%%%%%%%%%%%%%%%%%%%%%%%%%%%%%%%%%%%%%%%%%%%%%%%%%%%
Our analysis provides a control on the eigenfunctions of the operator $\opH$. Let $\vp_{\omega,E}$ denote the eigenfunction of $\opH$ corresponding to the eigenvalue $E$. The analysis of \cite[Theorem 8.6]{KLS} can be adapted to the continuum setting to show that
\begin{equation*}
	\lim_{n\to \infty} \frac{1}{n^{1-2\alpha}} \log \sqrt{|\vp_{\omega,E}(n)|^2 + |\vp'_{\omega,E}(n)|^2} = -\beta(\lambda,E), \quad \p-\text{a.s.},
\end{equation*}
for almost every fixed $E\in(0,\infty)$, where $\beta(\lambda,E)$ is explicitly given in \cite[Theorem 9.2]{KLS}(see Proposition \ref{thm:lyapunov} below). In particular, this shows that for almost every $E\in(0,\infty)$, $\p$-almost surely, there exists a finite constant $C_{\omega,E}$ such that
\begin{equation*}
	\sqrt{|\vp_{\omega,E}(x)|^2 + |\vp'_{\omega,E}(x)|^2} \leq C_{\omega,E}\ \e^{-\beta(\lambda,E) |x|^{1-2\alpha}}.
\end{equation*}
It is known that certain types of decay of eigenfunctions are closely related to dynamical localization \cite{DeRJLS1,DeRJLS2,GT}. Such criteria usually require a control on the localization centres of the eigenfunctions, uniformly in energy intervals. This information is missing in the above bound. We provide this uniform control in the next proposition.
%%%%%%%%%%%%%%%%%%%%%%%%%%%%%%%%%%%%%%%%%%%%%%%%%%%%%%%%%%%%
%%%%%%%%%%%%%%%%%%%%%%%%%%%%%%%%%%%%%%%%%%%%%%%%%%%%%%%%%%%%
\begin{theorem} \label{thm:decay-eigenfunctions}
	Let $0<\alpha<\frac12$ and let $\lambda\neq 0$.
	For all compact energy interval $I\subset (0,\infty)$, there exists two deterministic constants $c_1=c_1(I),\, c_2=c_2(I)$ and almost surely finite positive random quantities $c_{\omega}=c_{\omega}(I),\, C_{\omega}=C_{\omega}(I)$ such that
	\begin{eqnarray}\label{eq:SULE}
		c_{\omega} \, \e^{-c_1 |x|^{1-2\alpha}} 
		\leq
		\norm{ \chi_x  \vp_{\omega,E}}
		 \leq C_{\omega} \e^{-c_2 |x|^{1-2\alpha}}, 
		 %\quad \p-\text{a.s.},
		%\sqrt{|u_{\omega,E}(x)|^2 + |u'_{\omega,E}(x)|^2}
	\end{eqnarray}
	for all $E\in I\cap\sigma(\opH)$ and all $x\in\R$.
	%Furtheremore, we can take any $c_1 < \beta(\lambda,E)$.
\end{theorem}
%%%%%%%%%%%%%%%%%%%%%%%%%%%%%%%%%%%%%%%%%%%%%%%%%%%%%%%%%%%%
%%%%%%%%%%%%%%%%%%%%%%%%%%%%%%%%%%%%%%%%%%%%%%%%%%%%%%%%%%%%
The upper bound in \eqref{eq:SULE} can be seen as a stretched form of the condition SULE  where the localization centres are all equal to $0$ \cite{DeRJLS1,GK1}.

We finally state our result on the moments.
%%%%%%%%%%%%%%%%%%%%%%%%%%%%%%%%%%%%%%%%%%%%%%%%%%%%%%%%%%%%
%%%%%%%%%%%%%%%%%%%%%%%%%%%%%%%%%%%%%%%%%%%%%%%%%%%%%%%%%%%%
\begin{theorem}\label{thm:lower-bound-DL}
	Let $0<\alpha<\frac12$ and $\lambda\neq 0$, and let $I\subset \R$ be a compact interval. Then, 
	\begin{equation*}
		\esp\left(\sup_{t\in\R}
			\norm{
				\e^{\frac12 |X|^{\kappa}}\e^{-it\opH}P_I(\opH)\psi}^2
		\right)<\infty,
	\end{equation*}
	for all $\kappa<1-2\alpha$ and $\psi \in L^2(\R)$ with bounded support, while
	\begin{eqnarray*}
		\limsup_{t\to \infty} \norm{ \e^{\frac12|{X}|^{\kappa}}  \e^{-it \opH} \psi}^2 = \infty,
		\quad
		\p-a.s.,
	\end{eqnarray*}
	for all $\kappa>1-2\alpha$ and all $\psi \in {\text Ran}P_I(\opH)$.
\end{theorem}
%%%%%%%%%%%%%%%%%%%%%%%%%%%%%%%%%%%%%%%%%%%%%%%%%%%%%%%%%%%%
%%%%%%%%%%%%%%%%%%%%%%%%%%%%%%%%%%%%%%%%%%%%%%%%%%%%%%%%%%%%

%%%%%%%%%%%%%%%%%%%%%%%%%%%%%%%%%%%%%%%%%%%%%%%%%%%%%%%%%%%%
%%%%%%%%%%%%%%%%%%%%%%%%%%%%%%%%%%%%%%%%%%%%%%%%%%%%%%%%%%%%
%%%%%%%%%%%%%%%%%%%%%%%%%%%%%%%%%%%%%%%%%%%%%%%%%%%%%%%%%%%%
%%%%%%%%%%%%%%%%%%%%%%%%%%%%%%%%%%%%%%%%%%%%%%%%%%%%%%%%%%%%
%%%%%%%%%%%%%%%%%%%%%%%%%%%%%%%%%%%%%%%%%%%%%%%%%%%%%%%%%%%%
%%%%%%%%%%%%%%%%%%%%%%%%%%%%%%%%%%%%%%%%%%%%%%%%%%%%%%%%%%%%
%%%%%%%%%%%%%%%%%%%%%%%%%%%%%%%%%%%%%%%%%%%%%%%%%%%%%%%%%%%%
%%%%%%%%%%%%%%%%%%%%%%%%%%%%%%%%%%%%%%%%%%%%%%%%%%%%%%%%%%%%
%%%%%%%%%%%%%%%%%%%%%%%%%%%%%%%%%%%%%%%%%%%%%%%%%%%%%%%%%%%%
%%%%%%%%%%%%%%%%%%%%%%%%%%%%%%%%%%%%%%%%%%%%%%%%%%%%%%%%%%%%
%%%%%%%%%%%%%%%%%%%%%%%%%%%%%%%%%%%%%%%%%%%%%%%%%%%%%%%%%%%%
%%%%%%%%%%%%%%%%%%%%%%%%%%%%%%%%%%%%%%%%%%%%%%%%%%%%%%%%%%%%
%%%%%%%%%%%%%%%%%%%%%%%%%%%%%%%%%%%%%%%%%%%%%%%%%%%%%%%%%%%%
%%%%%%%%%%%%%%%%%%%%%%%%%%%%%%%%%%%%%%%%%%%%%%%%%%%%%%%%%%%%

\section{Asymptotics of Transfer Matrices and Pr\"ufer transform}\label{sec:transfer-matrices}

%%%%%%%%%%%%%%%%%%%%%%%%%%%%%%%%%%%%%%%%%%%%%%%%%%%%%%%%%%%%
%%%%%%%%%%%%%%%%%%%%%%%%%%%%%%%%%%%%%%%%%%%%%%%%%%%%%%%%%%%%
Let $\vp$ be a solution of the equation $\opH \vp = E \vp$ in some interval $[a,b]$. For $x,y \in [a,b]$,
we define the transfer matrices by the relation
\begin{eqnarray}\label{transfer matrix}
	\T(y,x;E) 
	\begin{pmatrix}
		\vp(x) \\ \vp'(x)
	\end{pmatrix}
	=
	\begin{pmatrix}
		\vp(y) \\ \vp'(y)
	\end{pmatrix}.
\end{eqnarray}
%%%%%%%%%%%%%%%%%%%%%%%%%%%%%%%%%%%%%%%%%%%%%%%%%%%%%%%%%%%%
%%%%%%%%%%%%%%%%%%%%%%%%%%%%%%%%%%%%%%%%%%%%%%%%%%%%%%%%%%%%
%%%%%%%%%%%%%%%%%%%%%%%%%%%%%%%%%%%%%%%%%%%%%%%%%%%%%%%%%%%%
%%%%%%%%%%%%%%%%%%%%%%%%%%%%%%%%%%%%%%%%%%%%%%%%%%%%%%%%%%%%
\begin{proposition}[\cite{KLS}, Theorem 9.2]\label{thm:lyapunov} Let $0<\alpha\leq \frac12$. For all $E\geq 0$ such that  $\sqrt{E}\notin \pi \Z$, we have the almost sure limit
	\begin{eqnarray}
		\beta(\lambda,E)
		:=
		\lim_{n\to\pm\infty} \frac{\log \| {\bf T}_{\omega}(n,0;E) \|}{\sum^n_{j=1}j^{-2\alpha}}
		=
		\frac{\lambda^2}{8E} 
		\left|
			\int^1_0 u(y)\ \e^{i \sqrt{E} y} dy
		\right|^2.
	\end{eqnarray}
\end{proposition}
%%%%%%%%%%%%%%%%%%%%%%%%%%%%%%%%%%%%%%%%%%%%%%%%%%%%%%%%%%%%
%%%%%%%%%%%%%%%%%%%%%%%%%%%%%%%%%%%%%%%%%%%%%%%%%%%%%%%%%%%%
%%%%%%%%%%%%%%%%%%%%%%%%%%%%%%%%%%%%%%%%%%%%%%%%%%%%%%%%%%%%
%%%%%%%%%%%%%%%%%%%%%%%%%%%%%%%%%%%%%%%%%%%%%%%%%%%%%%%%%%%%
%%%%%%%%%%%%%%%%%%%%%%%%%%%%%%%%%%%%%%%%%%%%%%%%%%%%%%%%%%%%
This result follows from the asymptotic analysis of the Pr\"ufer transform associated to the system. Following \cite{KLS},
we denote $k=\sqrt{E}$ and define the modified Pr\"ufer coordinates $R$ and $\theta$ such that
\begin{eqnarray}\label{prufer variables}
	\vp(x) &=& k R(x) \cos\theta(x),\notag\\
	\vp'(x) &=& R(x) \sin\theta(x).
\end{eqnarray}
Note that, if $V=0$, then we have $\theta(x)=\theta_0 + kx$. These satisfy the equations
\begin{eqnarray}
	%\label{eq:ODE-phase}
	\frac{d}{dx} \theta(x)
	&=&
	k - \frac{V_{\omega}(x)}{k} \sin^2 \theta(x),
	\\
	\label{eq:ODE-prufer}
	%\label{eq:ODE-radius}
	\frac{d}{dx} \log R(x)
	&=&
	\frac{1}{2k} V_{\omega}(x) \sin(2\theta(x)).
\end{eqnarray}
Note that the functions $R$ and $\theta$ depend on the energy $E$, which we removed from the notation as no confusion will arise. Nonetheless, we sometimes denote $R_{x_0}(\cdot;\theta_0)$ and $\theta_{x_0}(\cdot;\theta_0)$ to stress that the system is considered with initial conditions $\vp(x_0)=\sin \theta_0$ and $\vp'(x_0)=\cos \theta_0$.
%%%%%%%%%%%%%%%%%%%%%%%%%%%%%%%%%%%%%%%%%%%%%%%%%%%%%%%%%%%%
%%%%%%%%%%%%%%%%%%%%%%%%%%%%%%%%%%%%%%%%%%%%%%%%%%%%%%%%%%%%

We quote the following lemma from \cite{KLS}.
%%%%%%%%%%%%%%%%%%%%%%%%%%%%%%%%%%%%%%%%%%%%%%%%%%%%%%%%%%%%
%%%%%%%%%%%%%%%%%%%%%%%%%%%%%%%%%%%%%%%%%%%%%%%%%%%%%%%%%%%%
\begin{lemma}[\cite{KLS}, Lemma 2.1]\label{thm:comparison}
	For all compact energy interval $[a,b]\subset (0,\infty)$ and  all $\vartheta_1 \neq \vartheta_2$, there exists positive deterministic constants $C_1=C_1(\vartheta_1,\vartheta_2,I)$ and $C_2=C_2(\vartheta_1,\vartheta_2,I)$ such that
\begin{eqnarray}\label{eq:prufer-and-norms}
	C_1 \max\{ R_x(y,\vartheta_1),R_x(y,\vartheta_2)\}
	\leq 
	\| \T(E;y,x) \|
	\leq C_2 \max \{R_x(y,\vartheta_1),R_x(y,\vartheta_2)\}, 
\end{eqnarray}
for all $x,y\in\R$, $E\in I$ and all $\omega$. 
\end{lemma}
%%%%%%%%%%%%%%%%%%%%%%%%%%%%%%%%%%%%%%%%%%%%%%%%%%%%%%%%%%%%
%%%%%%%%%%%%%%%%%%%%%%%%%%%%%%%%%%%%%%%%%%%%%%%%%%%%%%%%%%%%
This allows to reduce the asymptotics of transfer matrices \eqref{transfer matrix} to the ones of the Pr\"ufer radii \eqref{prufer variables}. The analysis outlined in \cite[Section 9]{KLS} leads to
\begin{eqnarray}\label{eq:asymptotics-prufer}
	\esp\left[
		\log \frac{R(n)}{R(m)}
	\right]
	&=&
	\frac{\lambda^2 }{8k^2}
	\left|
	\int^1_0 u(y)\ 
	\e^{2iky}\, dy
	\right|^2
	\sum^n_{j=m} j^{-2\alpha}
	+
	K_{m,n},
\end{eqnarray}
for $m\leq n$, where $|K_{m,n}|=o(\sum^n_{j=m} j^{-2\alpha})$, uniformly on values of $\sqrt{E}\notin \pi \Z$ ranging over compact energy intervals. The same estimate holds for $n\leq m \leq 0$. By \eqref{eq:prufer-and-norms}, the same asymptotics holds for the norms of the transfer matrices. We detail the above estimate in Appendix \ref{app:martingale} and summarise it in a form that will suit our purposes in the next lemma.
%%%%%%%%%%%%%%%%%%%%%%%%%%%%%%%%%%%%%%%%%%%%%%%%%%%%%%%%%%%%
%%%%%%%%%%%%%%%%%%%%%%%%%%%%%%%%%%%%%%%%%%%%%%%%%%%%%%%%%%%%
For simplicity, we denote $T_{\omega,n}(E)=\T(n+1,n;E)$, dropping the dependence in $\lambda$.

%%%%%%%%%%%%%%%%%%%%%%%%%%%%%%%%%%%%%%%%%%%%%%%%%%%%%%%%%%%%
%%%%%%%%%%%%%%%%%%%%%%%%%%%%%%%%%%%%%%%%%%%%%%%%%%%%%%%%%%%%
\begin{lemma}\label{thm:bounds-on-Tmn}
	Let $I\subset(0,\infty)$ be a compact interval. Then for all $\beta'$ such that $0<\beta'<\displaystyle\inf_{E\in I}\beta(\lambda,E)$, there exists $n_0=n_0(I)\geq 1$ such that
	\begin{eqnarray}
		\label{eq:lower-bound-Tmn}
		\esp\left[ \log \|T_{\omega,ln_0 }(E)\cdots T_{\omega,(l-1)n_0+1}(E)\psi_0\|\right]
		&\geq&
		%\beta' \sum^{kn_0}_{j=(k-1)n_0+1} \frac{1}{j^{2\alpha}}
%		%+
%		%o\left(\frac{n_0^{1-2\alpha}}{k^{2\alpha}}\right).
		%\geq 
		\beta'
		\sum^{ln_0}_{j=(l-1)n_0+1} \frac{1}{j^{2\alpha}}
		 %\frac{n_0^{1-2\alpha}}{(1-2\alpha)l^{2\alpha}},
	\end{eqnarray}
	for all $l\geq 1$, $\|\psi_0\|=1$ and $E\in I$.
	\newline
	Furthermore, there exists a constant $C=C(I)$ such that
	\begin{eqnarray}
		\label{eq:upper-bound-Tmn}
		\esp\left[\left( \log \|T_{\omega,ln_0 }(E)\cdots T_{\omega,(l-1)n_0+1}(E)\psi_0\| \right)^2\right]
		&\leq&
		C
		\sum^{ln_0}_{j=(l-1)n_0+1} \frac{1}{j^{2\alpha}},
	\end{eqnarray}
	for all $l\geq 1$, $\|\psi_0\|=1$ and $E\in I$.
\end{lemma}
%%%%%%%%%%%%%%%%%%%%%%%%%%%%%%%%%%%%%%%%%%%%%%%%%%%%%%%%%%%%
%%%%%%%%%%%%%%%%%%%%%%%%%%%%%%%%%%%%%%%%%%%%%%%%%%%%%%%%%%%%
\begin{proof}
	From \eqref{eq:asymptotics-prufer}, we can find $n_0$ large enough such that
	\begin{eqnarray*}
		\esp\left[ \log \|T_{\omega,ln_0 }(E)\cdots T_{\omega,(l-1)n_0+1}(E)\psi_0\|\right]
		&\geq&
		\beta' \sum^{ln_0}_{j=(l-1)n_0+1} \frac{1}{j^{2\alpha}},
	\end{eqnarray*}
	for all $l\geq 1$, $\|\psi_0\|=1$ and all $E\in I$ corresponding to values of $k\notin \pi \Z$. The bound for all energies in $I$ then follows by continuity of the left-hand-side above with respect to $E$. This proves \eqref{eq:lower-bound-Tmn}. The upper bound \eqref{eq:upper-bound-Tmn} follows by an inspection of the martingale decomposition \eqref{eq:martingale-decomposition} in Appendix \ref{app:martingale}.
\end{proof}
%%%%%%%%%%%%%%%%%%%%%%%%%%%%%%%%%%%%%%%%%%%%%%%%%%%%%%%%%%%%
%%%%%%%%%%%%%%%%%%%%%%%%%%%%%%%%%%%%%%%%%%%%%%%%%%%%%%%%%%%%

%%%%%%%%%%%%%%%%%%%%%%%%%%%%%%%%%%%%%%%%%%%%%%%%%%%%%%%%%%%%
%%%%%%%%%%%%%%%%%%%%%%%%%%%%%%%%%%%%%%%%%%%%%%%%%%%%%%%%%%%%
%%%%%%%%%%%%%%%%%%%%%%%%%%%%%%%%%%%%%%%%%%%%%%%%%%%%%%%%%%%%
%%%%%%%%%%%%%%%%%%%%%%%%%%%%%%%%%%%%%%%%%%%%%%%%%%%%%%%%%%%%
%%%%%%%%%%%%%%%%%%%%%%%%%%%%%%%%%%%%%%%%%%%%%%%%%%%%%%%%%%%%
%%%%%%%%%%%%%%%%%%%%%%%%%%%%%%%%%%%%%%%%%%%%%%%%%%%%%%%%%%%%
%%%%%%%%%%%%%%%%%%%%%%%%%%%%%%%%%%%%%%%%%%%%%%%%%%%%%%%%%%%%
%%%%%%%%%%%%%%%%%%%%%%%%%%%%%%%%%%%%%%%%%%%%%%%%%%%%%%%%%%%%
%%%%%%%%%%%%%%%%%%%%%%%%%%%%%%%%%%%%%%%%%%%%%%%%%%%%%%%%%%%%
%%%%%%%%%%%%%%%%%%%%%%%%%%%%%%%%%%%%%%%%%%%%%%%%%%%%%%%%%%%%
%%%%%%%%%%%%%%%%%%%%%%%%%%%%%%%%%%%%%%%%%%%%%%%%%%%%%%%%%%%%
%%%%%%%%%%%%%%%%%%%%%%%%%%%%%%%%%%%%%%%%%%%%%%%%%%%%%%%%%%%%
%%%%%%%%%%%%%%%%%%%%%%%%%%%%%%%%%%%%%%%%%%%%%%%%%%%%%%%%%%%%
%%%%%%%%%%%%%%%%%%%%%%%%%%%%%%%%%%%%%%%%%%%%%%%%%%%%%%%%%%%%
%%%%%%%%%%%%%%%%%%%%%%%%%%%%%%%%%%%%%%%%%%%%%%%%%%%%%%%%%%%%
 
\section{Fractional moments estimates}\label{sec:FMM}

%%%%%%%%%%%%%%%%%%%%%%%%%%%%%%%%%%%%%%%%%%%%%%%%%%%%%%%%%%%%
%%%%%%%%%%%%%%%%%%%%%%%%%%%%%%%%%%%%%%%%%%%%%%%%%%%%%%%%%%%%
%%%%%%%%%%%%%%%%%%%%%%%%%%%%%%%%%%%%%%%%%%%%%%%%%%%%%%%%%%%%
%%%%%%%%%%%%%%%%%%%%%%%%%%%%%%%%%%%%%%%%%%%%%%%%%%%%%%%%%%%%
For $\Lambda\subset\R$, we denote by $H_{\omega,\Lambda}$ the restriction of $\opH$ to $L^2(\Lambda)$ and its resolvent by $G_{\omega,\Lambda}(E)=(H_{\omega,\Lambda}-E)^{-1}$, where we hid the explicit dependence on $\lambda$ to lighten the notation.
The following is the main result of this section.
%%%%%%%%%%%%%%%%%%%%%%%%%%%%%%%%%%%%%%%%%%%%%%%%%%%%%%%%%%%%
%%%%%%%%%%%%%%%%%%%%%%%%%%%%%%%%%%%%%%%%%%%%%%%%%%%%%%%%%%%%
\begin{theorem}\label{thm:FM}
	Let $0<\alpha<\frac12$ and $I\subset(0,\infty)$ be a compact interval.
	For each $y\in\Z$, there exists constants $c=c(I,y),\, C=C(I,y) \in (0,\infty)$ such that
	\begin{eqnarray*}
		\esp\left[ \norm{  \chi_x G_{\omega,[a,b]}(E) \chi_y}^s \right]
		\leq
		C (\lambda a_x)^{-1/2}\ \e^{-c|x|^{1-2\alpha}},
	\end{eqnarray*}
	for all $x\in\R, E\in I$ and all $a<b$.
\end{theorem}
%%%%%%%%%%%%%%%%%%%%%%%%%%%%%%%%%%%%%%%%%%%%%%%%%%%%%%%%%%%%
%%%%%%%%%%%%%%%%%%%%%%%%%%%%%%%%%%%%%%%%%%%%%%%%%%%%%%%%%%%%
The proof is given at the end of Section \ref{sec:estimates-transfer-matrices}. In Section \ref{sec:from-green-to-transfer-matrices}, we relate the fractional moments of the Green's function to negative fractional moments of the norm of transfer matrices which are then estimated in Section \ref{sec:estimates-transfer-matrices}.
%%%%%%%%%%%%%%%%%%%%%%%%%%%%%%%%%%%%%%%%%%%%%%%%%%%%%%%%%%%%
%%%%%%%%%%%%%%%%%%%%%%%%%%%%%%%%%%%%%%%%%%%%%%%%%%%%%%%%%%%%
%%%%%%%%%%%%%%%%%%%%%%%%%%%%%%%%%%%%%%%%%%%%%%%%%%%%%%%%%%%%
%%%%%%%%%%%%%%%%%%%%%%%%%%%%%%%%%%%%%%%%%%%%%%%%%%%%%%%%%%%%
%%%%%%%%%%%%%%%%%%%%%%%%%%%%%%%%%%%%%%%%%%%%%%%%%%%%%%%%%%%%
%%%%%%%%%%%%%%%%%%%%%%%%%%%%%%%%%%%%%%%%%%%%%%%%%%%%%%%%%%%%

\subsection{From Green's function to transfer matrices}\label{sec:from-green-to-transfer-matrices}

%%%%%%%%%%%%%%%%%%%%%%%%%%%%%%%%%%%%%%%%%%%%%%%%%%%%%%%%%%%%
%%%%%%%%%%%%%%%%%%%%%%%%%%%%%%%%%%%%%%%%%%%%%%%%%%%%%%%%%%%%
The following analysis is a direct adaptation of \cite[Section 3]{HSS}. 
%and allows us to estimate the fractional moments of the Green's function by negative moments of the norm of the transfer matrices. 
We provide the details for the sake of completeness and to carefully identify the dependence on the envelope $a_x$.

%%%%%%%%%%%%%%%%%%%%%%%%%%%%%%%%%%%%%%%%%%%%%%%%%%%%%%%%%%%%
%%%%%%%%%%%%%%%%%%%%%%%%%%%%%%%%%%%%%%%%%%%%%%%%%%%%%%%%%%%%
Fix $E\geq 0$ and let $I \subset \R$. For $c\in[a,b]$ and $\theta\in[0,2\pi)$, we define $\vp_c(\cdot;\theta)$ the solution of $H_{\omega,[a,b]} \vp = E\vp$ such that $\vp(c)=\sin \theta$ and $\vp'(c)=\cos \theta$. This way, we can define Pr\"ufer coordinates $R_c(x;\theta)$ and $\theta_c(x;\theta)$ with the convention that $\theta_c(c;\theta) = \theta$ and imposing continuity. We will eliminate $\theta$ from the notation whenever $\theta=0$.
%%%%%%%%%%%%%%%%%%%%%%%%%%%%%%%%%%%%%%%%%%%%%%%%%%%%%%%%%%%%
%%%%%%%%%%%%%%%%%%%%%%%%%%%%%%%%%%%%%%%%%%%%%%%%%%%%%%%%%%%%
\begin{lemma}\label{thm:from-green-to-transfer-matrices}
	Let $I\subset (0,\infty)$ be a compact inerval.
	For all $s\in[0,\frac12)$, there exists $C=C(s,I)\in(0,\infty)$ such that, for all $a<b$,
	\begin{eqnarray}
		\label{eq:from-green-to-transfer-matrices-1}
		\esp\left[ \norm{  \chi_x G_{\omega,[a,b]}(E) \chi_y }^s \right]
		\leq C \, (\lambda a_x)^{-1/2}
		\esp \left[ \left\| 
			\T(E;x,y)
			\begin{pmatrix}
				\sin \theta_b(y)
				\\
				\cos \theta_b(y)
			\end{pmatrix}
		\right\|^{-2s} \right]^{1/2},
	\end{eqnarray}
	for all integers $a \leq x < y \leq b$ and
	\begin{eqnarray}
		\label{eq:from-green-to-transfer-matrices-2}
		\esp\left[ \norm{ \chi_x G_{\omega,[a,b]}(E) \chi_y}^s \right]
		\leq C \, (\lambda a_x)^{-1/2}
		\esp \left[ \left\| 
			\T(E;x,y)
			\begin{pmatrix}
				\sin \theta_a(y)
				\\
				\cos \theta_a(y)
			\end{pmatrix}
		\right\|^{-2s} \right]^{1/2},
	\end{eqnarray}
	for all integers $a \leq y < x \leq b$.
\end{lemma}
%%%%%%%%%%%%%%%%%%%%%%%%%%%%%%%%%%%%%%%%%%%%%%%%%%%%%%%%%%%%
%%%%%%%%%%%%%%%%%%%%%%%%%%%%%%%%%%%%%%%%%%%%%%%%%%%%%%%%%%%%
\begin{proof}
	We start from the identity
	\begin{eqnarray*}
		G_{\omega,[a,b]}(s,t;E)
		=
		\frac{1}{W(\vp_a,\vp_b)}
		\left\{
			\begin{array}{ll}
				\vp_a(s)\ \vp_b(t) & \text{if} \, s\leq t, 
				\\
				\vp_a(t)\ \vp_b(s) & \text{if} \, s > t,
			\end{array}
		\right.
	\end{eqnarray*}
	where $W(f,g)=fg'-f'g$ is the Wronskian of the functions $f$ and $g$. We consider $a \leq x < y \leq b$ as the opposite case follows by symmetry.
	Note that $W(\vp_a,\vp_b)(x) = k R_a(x) R_b(x) \sin(\theta_a(x)-\theta_b(x))$. Hence, by definition of the Pr\"ufer transform, we can find a constant $C=C(I)\in(0,\infty)$ such that
	\begin{eqnarray*}
		\esp\left[ \norm{  \chi_x G_{\omega,[a,b]}(E) \chi_y }^s \right]
		&\leq&
		C \,
		\esp\left[
			\frac{R_b(y)^s}{R_b(x)^s} 
			| \sin\left( \theta_a(x)-\theta_b(x)\right)|^{-s}
		\right]
		\\
		&\leq&
		C \,
		\esp\left[
			\frac{R_b(y)^{2s}}{R_b(x)^{2s}} 
		\right]^{1/2}
		\esp\left[
			| \sin\left( \theta_a(x)-\theta_b(x)\right)|^{-2s}
		\right]^{1/2}.
	\end{eqnarray*}
	From \eqref{eq:single-site-potential}, we infer that $R_b(x)=R_b(y)R_y(x;\theta_b(y))$. Hence,
\begin{eqnarray*}
		\esp\left[ \norm{  \chi_x G_{\omega,[a,b]}(E) \chi_y }^s \right]
%		&\leq&
%		C
%		\esp\left[
%			R_y(x, \theta_b(y))^{-s}
%			| \sin\left( \theta_a(x)-\theta_b(x)\right)|^{-s}
%		\right]
%		\\
		&\leq&
		C \,
		\esp\left[
			R_y(x, \theta_b(y))^{-2s}  
		\right]^{1/2}
		\esp\left[
			| \sin\left( \theta_a(x)-\theta_b(x)\right)|^{-2s}
		\right]^{1/2}.
	\end{eqnarray*}
	The first expected value above is bounded by the expected value on the right hand side of \eqref{eq:from-green-to-transfer-matrices-1}. The bound on the second one is given in the next lemma. 
%	An inspection of Lemma 3.2 in \cite{HSS} shows that
%	\begin{eqnarray}
%		\esp\left[
%			| \sin\left( \theta_a(x)-\theta_b(x)\right)|^{-2s}
%		\right]
%		\leq
%		C \| \rho_x\|_{\infty},
%	\end{eqnarray}
%	where $\rho_x$ is the density of $\lambda a_x \omega_x$. As $\rho_x(y)=\lambda^{-1} a_x^{-1} \rho(\lambda^{-1}a_x^{-1} y)$, the result follows.
\end{proof}
%%%%%%%%%%%%%%%%%%%%%%%%%%%%%%%%%%%%%%%%%%%%%%%%%%%%%%%%%%%%
%%%%%%%%%%%%%%%%%%%%%%%%%%%%%%%%%%%%%%%%%%%%%%%%%%%%%%%%%%%%

%%%%%%%%%%%%%%%%%%%%%%%%%%%%%%%%%%%%%%%%%%%%%%%%%%%%%%%%%%%%
%%%%%%%%%%%%%%%%%%%%%%%%%%%%%%%%%%%%%%%%%%%%%%%%%%%%%%%%%%%%
\begin{lemma}
	For all compact energy interval $[a,b]\subset (0,\infty)$
	and for all $s\in[0,\frac12)$, there exists $C=C(s,I)\in(0,\infty)$ such that
	\begin{eqnarray*}
		\int^{\omega^+}_{\omega^-} \frac{\rho(\omega_x)d\omega_x}{|\sin(\theta_b(x)-\theta_a(x))|^{2s}}
		\leq 
		C (\lambda a_x)^{-1} \| \rho \|,
	\end{eqnarray*}
	for all $a<x<b$ and all realization of $\{\omega_y:\, y\neq x\}$.
\end{lemma}
%%%%%%%%%%%%%%%%%%%%%%%%%%%%%%%%%%%%%%%%%%%%%%%%%%%%%%%%%%%%
%%%%%%%%%%%%%%%%%%%%%%%%%%%%%%%%%%%%%%%%%%%%%%%%%%%%%%%%%%%%
\begin{proof}
	Observe that $\theta_a(x)$ is independent of $\omega_x$. We will change variables to $t=\theta_b(x)$. Now, from \eqref{eq:ODE-prufer}, we see that, for $s\in[x,x+1)$, we have
	\begin{eqnarray*}
		\frac{d}{ds}\theta_a(s)
		&=&
		k - \frac{\lambda a_x \omega_x}{k} u_x(s)\sin \theta_a(s),
		\\
		\frac{d}{ds} \log R_a(s)
		&=&
		\frac{\lambda a_x \omega_x}{2k}u_x(s)\sin 2\theta_a(s).
	\end{eqnarray*}
	Hence,
	\begin{eqnarray}\label{deriv y theta}
		\frac{\partial}{\partial y} 
		\left(
			R_a(s)^2 \frac{\partial}{\partial \omega_x} \theta_a(s)
		\right)
		=
		- \frac{\lambda a_x}{k}u_x(s) R_a(s)^2 \sin^2\theta_a(s)
		=
		- \frac{\lambda a_x}{k^2}u_x(s)\vp_a(s)^2.
	\end{eqnarray}
	Since $\theta_b(x+1)$ is independent of $\omega_x$, \eqref{deriv y theta} yields
	\begin{eqnarray}
		\frac{\partial}{\partial \omega_x} \theta_a(x)
		=
		\frac{\lambda a_x}{k^2 R_a(x)^2}
		\int^{x+1}_x
		u_x(s)\vp_a(s)^2
		ds.
	\end{eqnarray}
	From \eqref{eq:single-site-potential}, Lemma \ref{thm:general-bound-1} and \ref{thm:general-bound-2}, we can then find two positive constants $C_1$ and $C_2$ such that
	\begin{eqnarray}
		C_1 \lambda a_x
		\leq
		\left| \frac{\partial}{\partial \omega_x} \theta_a(x) \right|
		\leq
		C_2 \lambda a_x.
	\end{eqnarray}
	This change of variables leads to the estimate.
\end{proof}
%%%%%%%%%%%%%%%%%%%%%%%%%%%%%%%%%%%%%%%%%%%%%%%%%%%%%%%%%%%%
%%%%%%%%%%%%%%%%%%%%%%%%%%%%%%%%%%%%%%%%%%%%%%%%%%%%%%%%%%%%

%%%%%%%%%%%%%%%%%%%%%%%%%%%%%%%%%%%%%%%%%%%%%%%%%%%%%%%%%%%%
%%%%%%%%%%%%%%%%%%%%%%%%%%%%%%%%%%%%%%%%%%%%%%%%%%%%%%%%%%%%
%%%%%%%%%%%%%%%%%%%%%%%%%%%%%%%%%%%%%%%%%%%%%%%%%%%%%%%%%%%%
%%%%%%%%%%%%%%%%%%%%%%%%%%%%%%%%%%%%%%%%%%%%%%%%%%%%%%%%%%%%
%%%%%%%%%%%%%%%%%%%%%%%%%%%%%%%%%%%%%%%%%%%%%%%%%%%%%%%%%%%%
%%%%%%%%%%%%%%%%%%%%%%%%%%%%%%%%%%%%%%%%%%%%%%%%%%%%%%%%%%%%

\subsection{Estimates on transfer matrices}\label{sec:estimates-transfer-matrices}

%%%%%%%%%%%%%%%%%%%%%%%%%%%%%%%%%%%%%%%%%%%%%%%%%%%%%%%%%%%%
%%%%%%%%%%%%%%%%%%%%%%%%%%%%%%%%%%%%%%%%%%%%%%%%%%%%%%%%%%%%
We start with an a priori estimate on the norm of transfer matrices.
%%%%%%%%%%%%%%%%%%%%%%%%%%%%%%%%%%%%%%%%%%%%%%%%%%%%%%%%%%%%
%%%%%%%%%%%%%%%%%%%%%%%%%%%%%%%%%%%%%%%%%%%%%%%%%%%%%%%%%%%%
\begin{lemma}\label{thm:transfer-matrices-a-priori}
	For all compact interval $I\subset (0,\infty)$ and for all $[a,b]\subset \R$, there exists $M=M(I,a,b) \in (0,\infty)$ such that
	\begin{eqnarray*}
		M^{-1}
		\leq
		\| \T(x,y;E)\|
		\leq
		M,
	\end{eqnarray*}
	for all $E\in I$ and $x,y \in [a,b]$.
\end{lemma}
%%%%%%%%%%%%%%%%%%%%%%%%%%%%%%%%%%%%%%%%%%%%%%%%%%%%%%%%%%%%
%%%%%%%%%%%%%%%%%%%%%%%%%%%%%%%%%%%%%%%%%%%%%%%%%%%%%%%%%%%%
\begin{proof}
	The estimates of Lemma \ref{thm:general-bound-1} from Appendix \ref{app:general} imply that
	\begin{eqnarray*}
		\exp \left( -\frac12 \int^b_a |1+V(t)-E|\, dt \right)
		\leq
		\| \T(x,y;E)\| 
		\leq
		\exp \left( \frac12 \int^b_a |1+V(t)-E|\, dt \right).
	\end{eqnarray*}
\end{proof}

The proofs of the next two lemmas are strongly inspired by \cite{CKM} but we provide them in full details as our non ergodic situation requires finer estimates. For applications of this argument in the continuum ergodic setting, see \cite{HSS, DaSiSt}. Recall the notation $T_{\omega,n}(E)=\T(n+1,n;E)$.
%%%%%%%%%%%%%%%%%%%%%%%%%%%%%%%%%%%%%%%%%%%%%%%%%%%%%%%%%%%%
%%%%%%%%%%%%%%%%%%%%%%%%%%%%%%%%%%%%%%%%%%%%%%%%%%%%%%%%%%%%
\begin{lemma}\label{thm:first-bound-Tmn}
	Let $0<\alpha< \frac12$ and $\lambda\neq 0$. For all compact interval $I\subset (0,\infty)$, there exist
	$n_0=n_0(I)\geq 1$, $s_0=s_0(I)\in(0,1)$ and $c=c(I)>0$ such that
	\begin{eqnarray*}
		\esp\left[ \|T_{\omega,ln_0 }(E)\cdots T_{\omega,(l-1)n_0+1}(E)\psi_0\|^{-s} \right]
		\leq
		1 - \frac{c}{l^{2\alpha}},
	\end{eqnarray*}
	for all $s\in(0,s_0]$, $l\geq 1$, $\|\psi_0\|=1$ and $E\in I$.
\end{lemma}
%%%%%%%%%%%%%%%%%%%%%%%%%%%%%%%%%%%%%%%%%%%%%%%%%%%%%%%%%%%%
%%%%%%%%%%%%%%%%%%%%%%%%%%%%%%%%%%%%%%%%%%%%%%%%%%%%%%%%%%%%
\begin{proof}
	We drop the dependence on $E$ to lighten the notation.
	From Lemma \ref{thm:bounds-on-Tmn}, we obtain $n_0=n_0(I)\geq 1$, $c_1=c_1(I)>0$ and $c_2=c_2(I)>0$ such that
	\begin{eqnarray*}
		\esp\left[ \log \|T_{\omega,ln_0 }\cdots T_{\omega,(l-1)n_0+1}\psi_0\|\right]
		\geq c_1
		 \frac{n_0^{1-2\alpha}}{l^{2\alpha}},
	\end{eqnarray*}
	and
	\begin{eqnarray*}
		\esp\left[
			\left(
				\log \|T_{\omega,ln_0 }\cdots T_{\omega,(l-1)n_0+1}\psi_0\|
			\right)^2
		\right]
		\leq c_2
		\frac{n_0^{1-2\alpha}}{l^{2\alpha}},
	\end{eqnarray*}
	for all $l\geq 1$, $\|\psi_0\|=1$ and $E\in I$.
	Now, we apply the inequality $\e^y\leq 1 + y + y^2 \e^{|y|}$ to $y=-s \log \|T_{\omega,ln_0} \cdots T_{\omega,(l-1)n_0+1}\psi_0\|$ with $s$ to be fixed later, so that
	\begin{eqnarray*}
		&&\|T_{\omega,ln_0 }\cdots T_{\omega,(l-1)n_0+1}\psi_0\|^{-s} 
		\leq 
		1- s \log \|T_{\omega,ln_0 }\cdots T_{\omega,(l-1)n_0+1}\psi_0\|
		\\
		&&\quad + \quad
		s^2 \left( \log \|T_{\omega,ln_0 }\cdots T_{\omega,(l-1)n_0+1}\psi_0\|\right)^2
		\e^{s \left|\log \|T_{\omega,ln_0 }\cdots T_{\omega,(l-1)n_0+1}\psi_0\| \right|}.
	\end{eqnarray*}
	Now,
	\begin{eqnarray*}
		\log \norm{T_{\omega,ln_0 }\cdots T_{\omega,(l-1)n_0+1}\psi_0} \leq \sum^{ln_0}_{j=(l-1)n_0+1}\log \norm{ T_{\omega,j}}.
	\end{eqnarray*}
	On the other hand,
	\begin{eqnarray*}
		1 &=& \| T_{\omega,(j-1)n_0+1}^{-1}\cdots T_{\omega,ln_0}^{-1} T_{\omega,ln_0}\cdots T_{\omega,(l-1)n_0+1} \psi_0\| \\
			&\leq&
			\| T_{\omega,(l-1)n_0+1} \| \cdots \| T_{\omega,ln_0}\| \| T_{\omega,ln_0}\cdots T_{\omega,(l-1)n_0+1} \psi_0\|,
	\end{eqnarray*}
	since $\| T_{\omega,j}^{-1}\| = \| T_{\omega,j}\|$, so that we have
	\begin{eqnarray*}
		\log \|T_{\omega,ln_0 }\cdots T_{\omega,(l-1)n_0+1}\psi_0\| 
		\geq 
		-\sum^{ln_0}_{j=(l-1)n_0+1}\log \| T_{\omega,j}\|.
	\end{eqnarray*}
	Piecing these bounds together and remembering Lemma \ref{thm:transfer-matrices-a-priori}, we obtain
	\begin{eqnarray*}
		%\esp\left[
		\left| \log \norm{ T_{\omega,ln_0 }\cdots T_{\omega,(l-1)n_0+1}\psi_0} \right| 
		%\right]
		\leq 
		\sum^{ln_0}_{j=(l-1)n_0+1} 
		%\esp\left[
		\log \| T_{\omega,j}\| %\right] 
		\leq c_3 n_0,
	\end{eqnarray*}
	for some $c_3=c_3(I)>0$.
	Hence, 
	\begin{eqnarray*}
		\esp\left[ \norm{T_{\omega,ln_0 }\cdots T_{\omega,(l-1)n_0+1}\psi_0}^{-s} \right]
		\leq
		1 - c_1 s \frac{n_0^{1-2\alpha}}{l^{2\alpha}} 
		+
		c_2 s^2 \e^{c_3 n_0} \frac{n_0^{1-2\alpha}}{l^{2\alpha}},
	\end{eqnarray*}
	for all $l\ge 1$, $\norm{\psi_0}=1$ and $E\in I$. We can now find  $s_0=s_0(I)>0$ small enough such that
	\begin{eqnarray*}
		\esp\left[ \norm{T_{\omega,ln_0 }\cdots T_{\omega,(l-1)n_0+1}\psi_0}^{-s} \right]
		\leq
		1 - \frac{c_4}{l^{2\alpha}},
	\end{eqnarray*}
	for some $c_4>0$, for all $s\in(0,s_0]$, $l\geq 1$, $\norm{\psi_0}=1$ and $E\in I$. 
\end{proof}
%%%%%%%%%%%%%%%%%%%%%%%%%%%%%%%%%%%%%%%%%%%%%%%%%%%%%%%%%%%%
%%%%%%%%%%%%%%%%%%%%%%%%%%%%%%%%%%%%%%%%%%%%%%%%%%%%%%%%%%%%
The following lemma contains the key estimate to Theorem \ref{thm:FM}.
%%%%%%%%%%%%%%%%%%%%%%%%%%%%%%%%%%%%%%%%%%%%%%%%%%%%%%%%%%%%
%%%%%%%%%%%%%%%%%%%%%%%%%%%%%%%%%%%%%%%%%%%%%%%%%%%%%%%%%%%%
\begin{lemma}\label{thm:key-lemma}
Let $0<\alpha< \frac12$ and $\lambda\neq 0$. For each $m\in\Z$ and each compact interval $I\subset  (0,\infty)$, there exist  $s_0=s_0(m,I)\in(0,1)$, $C=C(m,I)$ and $c=c(I)>0$ such that
	\begin{eqnarray}\label{eq:bound-key-lemma}
		\esp\left[
			\norm{
			{\bf T}_{\omega}(n,m;E)
			\psi_0}^{-s} 
			\right] 
			\leq 
			C\e^{-c|n|^{1-2\alpha}},
	\end{eqnarray}
	for all $s\in(0,s_0]$, $\|\psi_0\|=1$ and $n\in\Z$.
\end{lemma}
%%%%%%%%%%%%%%%%%%%%%%%%%%%%%%%%%%%%%%%%%%%%%%%%%%%%%%%%%%%%
%%%%%%%%%%%%%%%%%%%%%%%%%%%%%%%%%%%%%%%%%%%%%%%%%%%%%%%%%%%%
\begin{proof}
	Once again, we drop the dependence on $E$ to lighten the notation.
	We start proving the bound \eqref{eq:bound-key-lemma} for $m\geq 0$, $n\geq m + n_0$ and $s\in(0,s_0]$ where $n_0=n_0(I)\geq 1$ and $s_0=s_0(I)>0$ are taken from the previous lemma.
	Write $m= l_1n_0 - r_1$ and $n=l_2n_0 + r_2$ with $0\leq r_1,\, r_2 < n_0$.
	By Lemma \ref{thm:transfer-matrices-a-priori},
	\begin{eqnarray*}
		\norm{ T_{\omega,l_2n_0} \cdots T_{\omega,m} \psi_0 }
		&=&
		\|T_{\omega,l_2n_0+1}^{-1} T_{\omega,n}^{-1}T_{\omega,n}\cdots T_{\omega,m} \psi_0\|\\
		&\leq&
		\prod^n_{j=l_2n_0+1}\norm{ T_{\omega,j}} 
		\cdot
		\norm{ T_{\omega,n} \cdots T_{\omega,m} \psi_0}\\
		&\leq&
		C_1 \norm{ T_{\omega,n} \cdots T_{\omega,m} \psi_0},
	\end{eqnarray*}
	for some $C_1=C_1(I)>0$. 
	The rest of the proof is  based on a careful conditioning that we now detail.	
	Let
	\begin{eqnarray*}
		\psi_{l} = \frac{T_{ln_0} \cdots T_{m} \psi_0}{\norm{ T_{ln_0} \cdots T_{m} \psi_0}},
	\end{eqnarray*}
	and observe that $\psi_{l-1}$ is measurable with respect to $\mathcal{F}_{l-1}$.
	Hence, Lemma \ref{thm:first-bound-Tmn} can be applied to obtain
	\begin{eqnarray*}
		\esp
		\left[\norm{T_{\omega,ln_0} \cdots T_{\omega,(l-1)n_0 +1}  \psi_{l-1}}^{-s}
		\Big{|}
		\mathcal{F}_{l-1}
		\right]
		\leq
		1-\frac{c_4}{l^{2\alpha}},
	\end{eqnarray*}
	where $c_4=c_4(I)>0$ is the constant from Lemma \ref{thm:first-bound-Tmn}. Hence,
	\begin{eqnarray*}
		&&
		\esp\left[\norm{ T_{\omega,n} \cdots T_{\omega,m} \psi_0}^{-s}
		\right]
		\leq
		C_1
		\esp\left[\norm{T_{\omega,l_2n_0} \cdots T_{\omega,m} \psi_0}^{-s}
		\right]
		\\
		&&
		=
		C_1
		\esp\left[\norm{ T_{\omega,(l_2-1)n_0} \cdots T_{\omega,m} \psi_0}^{-s}
			\norm{ T_{\omega,l_2n_0} \cdots T_{\omega,(l_2-1)n_0+1} \psi_{l_2-1}}^{-s} 
		\right]
		\\
		&&
		=
		C_1
		\esp\left[ \esp\left[\norm{T_{\omega,(l_2-1)n_0} \cdots T_{\omega,m} \psi_0}^{-s}
			\norm{ T_{\omega,l_2n_0} \cdots T_{(l_2-1)n_0+1} \psi_{l_2-1}}^{-s} 
		\Big{|}
		\mathcal{F}_{l_2-1}
		\right]\right]
		\\
		&&
		=
		C_1
		\esp\left[ [\norm{ T_{\omega,(l_2-1)n_0} \cdots T_{\omega,m} \psi_0}^{-s}
			 \esp\left[\norm{T_{\omega,l_2n_0} \cdots T_{\omega,(l_2-1)n_0+1} \psi_{l_2-1}}^{-s} 
		\Big{|}
		\mathcal{F}_{l_2-1}
		\right]\right]
		\\
		&&
		\leq
		C_1 
		\left( 1-\frac{c_4}{l_2^{2\alpha}}\right)
		\esp\left[\norm{T_{\omega,(l_2-1)n_0} \cdots T_{\omega,m} \psi_0}^{-s}
		\right].
	\end{eqnarray*}
	Iterating, we get
	\begin{eqnarray*}
		\esp\left[\norm{ T_{\omega,n} \cdots T_{\omega,m} \psi_0}^{-s}
		\right]
		&\leq&
		C_1
		\prod^{l_2}_{j=l_1} \left( 1-\frac{c_4}{j^{2\alpha}}\right)
		\esp\left[\norm{ T_{\omega,m+r_1} \cdots T_{\omega,m} \psi_0}^{-s}
		\right].
	\end{eqnarray*}
	Just as we did in the previous lemma, we have
	\begin{eqnarray*}
		1 = \norm{T_{\omega,m}^{-1} \cdots T_{\omega,m+r_1}^{-1} T_{\omega,m+r_1} \cdots T_{\omega,m}\psi_0}
		\leq 
		\prod^{m+r_1}_{j=m} \norm{ T_{\omega,j} \| \cdot  \|T_{\omega,m+r_1} \cdots T_{\omega,m}\psi_0},
	\end{eqnarray*}
	so that, by Lemma \ref{thm:transfer-matrices-a-priori},
	\begin{eqnarray*}
		\esp\left[\norm{T_{\omega,m+r_1} \cdots T_{\omega,m}\psi_0}^{-s}
		\right]
		\leq C_2,
	\end{eqnarray*}
	for some $C_2=C_2(I)>0$.
	Hence,
	\begin{eqnarray*}
		\esp\left[\norm{ T_{\omega,n} \cdots T_{\omega,m} \psi_0}^{-s}
		\right]
		\leq
		C_3
		\prod^{l_2}_{j=l_1} \left( 1-\frac{c_4}{j^{2\alpha}}\right)
		\leq
		C_3\
		\e^{cm^{1-2\alpha}-cn^{1-2\alpha}},
	\end{eqnarray*}
	for some suitable $C_3=C_3(I)>0$ and $c=c(I)>0$.
	\newline
	The symmetric situation where $m\leq 0$ and $n\leq m-n_0$ is treated in the exact same way. If $m\leq 0 $ and $n\geq n_0$, the analysis is essentially reduced to estimate $\T(n,0)$. Indeed, we notice that
	\begin{eqnarray*}
		\T(n,0)=T^{-1}_{\omega,m}\cdots T^{-1}_{\omega,-1}\T(n,m),
	\end{eqnarray*}
	so that, by Lemma \ref{thm:transfer-matrices-a-priori},
	\begin{eqnarray*}
		\norm{ \T(n,m)\|^{-s} \leq C_4 \| \T(n,0)}^{-s},
	\end{eqnarray*}
	for some $C_4=C_4(I)>0$. The right-hand-side is covered by our previous discussion.
	The case $m\geq 0$ and $n\leq -n_0$ is of course similar. 
	In all the remaining cases, we simply use the a priori bound of Lemma \ref{thm:transfer-matrices-a-priori} so that
	\begin{eqnarray*}
		\esp\left[
			\norm{	\T(n,m)
			\psi_0
			}^{-s} 
			\right] 
		\leq 
		\esp\left[
			\norm{\T(n,m)}^{s} 
			\right] 
		\leq
		M^{|n-m|},
	\end{eqnarray*}
	where $M=M(I)$ was defined in Lemma \ref{thm:transfer-matrices-a-priori}.	
\end{proof}
%%%%%%%%%%%%%%%%%%%%%%%%%%%%%%%%%%%%%%%%%%%%%%%%%%%%%%%%%%%%
%%%%%%%%%%%%%%%%%%%%%%%%%%%%%%%%%%%%%%%%%%%%%%%%%%%%%%%%%%%%
%%%%%%%%%%%%%%%%%%%%%%%%%%%%%%%%%%%%%%%%%%%%%%%%%%%%%%%%%%%%
%%%%%%%%%%%%%%%%%%%%%%%%%%%%%%%%%%%%%%%%%%%%%%%%%%%%%%%%%%%%
\begin{proof}[Proof of Theorem \ref{thm:FM}]
	Use Lemma \ref{thm:from-green-to-transfer-matrices} to bound the fractional moments of the Green's function by the negative fractional moments of the norm of transfer matrices. These can be estimated by Lemma \ref{thm:key-lemma}.
\end{proof}
%%%%%%%%%%%%%%%%%%%%%%%%%%%%%%%%%%%%%%%%%%%%%%%%%%%%%%%%%%%%
%%%%%%%%%%%%%%%%%%%%%%%%%%%%%%%%%%%%%%%%%%%%%%%%%%%%%%%%%%%%

%%%%%%%%%%%%%%%%%%%%%%%%%%%%%%%%%%%%%%%%%%%%%%%%%%%%%%%%%%%%
%%%%%%%%%%%%%%%%%%%%%%%%%%%%%%%%%%%%%%%%%%%%%%%%%%%%%%%%%%%%
%%%%%%%%%%%%%%%%%%%%%%%%%%%%%%%%%%%%%%%%%%%%%%%%%%%%%%%%%%%%
%%%%%%%%%%%%%%%%%%%%%%%%%%%%%%%%%%%%%%%%%%%%%%%%%%%%%%%%%%%%
%%%%%%%%%%%%%%%%%%%%%%%%%%%%%%%%%%%%%%%%%%%%%%%%%%%%%%%%%%%%
%%%%%%%%%%%%%%%%%%%%%%%%%%%%%%%%%%%%%%%%%%%%%%%%%%%%%%%%%%%%
%%%%%%%%%%%%%%%%%%%%%%%%%%%%%%%%%%%%%%%%%%%%%%%%%%%%%%%%%%%%
%%%%%%%%%%%%%%%%%%%%%%%%%%%%%%%%%%%%%%%%%%%%%%%%%%%%%%%%%%%%
%%%%%%%%%%%%%%%%%%%%%%%%%%%%%%%%%%%%%%%%%%%%%%%%%%%%%%%%%%%%
%%%%%%%%%%%%%%%%%%%%%%%%%%%%%%%%%%%%%%%%%%%%%%%%%%%%%%%%%%%%
%%%%%%%%%%%%%%%%%%%%%%%%%%%%%%%%%%%%%%%%%%%%%%%%%%%%%%%%%%%%
%%%%%%%%%%%%%%%%%%%%%%%%%%%%%%%%%%%%%%%%%%%%%%%%%%%%%%%%%%%%
%%%%%%%%%%%%%%%%%%%%%%%%%%%%%%%%%%%%%%%%%%%%%%%%%%%%%%%%%%%%
%%%%%%%%%%%%%%%%%%%%%%%%%%%%%%%%%%%%%%%%%%%%%%%%%%%%%%%%%%%%
%%%%%%%%%%%%%%%%%%%%%%%%%%%%%%%%%%%%%%%%%%%%%%%%%%%%%%%%%%%%
%%%%%%%%%%%%%%%%%%%%%%%%%%%%%%%%%%%%%%%%%%%%%%%%%%%%%%%%%%%%
%%%%%%%%%%%%%%%%%%%%%%%%%%%%%%%%%%%%%%%%%%%%%%%%%%%%%%%%%%%%
%%%%%%%%%%%%%%%%%%%%%%%%%%%%%%%%%%%%%%%%%%%%%%%%%%%%%%%%%%%%

\section{Proof of Dynamical Localization}\label{sec:DL}

%%%%%%%%%%%%%%%%%%%%%%%%%%%%%%%%%%%%%%%%%%%%%%%%%%%%%%%%%%%%
%%%%%%%%%%%%%%%%%%%%%%%%%%%%%%%%%%%%%%%%%%%%%%%%%%%%%%%%%%%%
We outline the theory developed in \cite[Section 2]{AENSS} to relate the fractional moments of the Green's function to the correlator \eqref{eq:correlator}, with some one-dimensional adaptations from \cite{HSS}. Since our potential is not ergodic, some care has to be taken to insure that the estimates remain uniform enough. We provide details when this is required. The proof of Theorem \ref{thm:DL} is given at the end of the section.

%%%%%%%%%%%%%%%%%%%%%%%%%%%%%%%%%%%%%%%%%%%%%%%%%%%%%%%%%%%%
%%%%%%%%%%%%%%%%%%%%%%%%%%%%%%%%%%%%%%%%%%%%%%%%%%%%%%%%%%%%
To simplify the notation, we denote $H_{\omega,L}=H_{\omega,[-L,L]}$ the restriction of $\opH$ to the box $[-L,L]$, and we let
$G_{\omega,L}(E)=(H_{\omega,L}-E)^{-1}$ its resolvent.
%Recall that, for $\Lambda\subset\R$, we denote by $H_{\omega,\Lambda}$ the restriction of $\opH$ to $L^2(\Lambda)$ and we let $G_{\omega,\Lambda}(E)=(H_{\omega,\Lambda}-E)^{-1}$.
For an energy interval $I$, we consider the restricted correlator to the box $[-L,L]$ defined by
\begin{eqnarray}\label{eq:correlator-box}
		Q_{\omega,L}(x,y;I)
	=\sup_{\substack{f\in \mathcal{C}_{c}\\ \norm{f}\le1}}
	\norm{ \chi_x f(H_{\omega,L}) \chi_y},
\end{eqnarray}
where, once again, we dropped the dependence on $\lambda$ to lighten the notation.
%%%%%%%%%%%%%%%%%%%%%%%%%%%%%%%%%%%%%%%%%%%%%%%%%%%%%%%%%%%%
%%%%%%%%%%%%%%%%%%%%%%%%%%%%%%%%%%%%%%%%%%%%%%%%%%%%%%%%%%%%
The following is the main result of this section.
%%%%%%%%%%%%%%%%%%%%%%%%%%%%%%%%%%%%%%%%%%%%%%%%%%%%%%%%%%%%
%%%%%%%%%%%%%%%%%%%%%%%%%%%%%%%%%%%%%%%%%%%%%%%%%%%%%%%%%%%%
\begin{theorem}\label{thm:from-green-to-correlator}
	Let $0<v<s<1$ and let $I\subset (0,\infty)$ be a compact interval. Then, there exists a constant $C=C(v,s,I)$ such that
	\begin{eqnarray}
		\esp\left[ Q_{\omega,L}(x,m;I)\right]
		\leq 
		C (\lambda a_m)^{\frac{v-1}{2-v}} \| \rho \|_{\infty}^{\frac{1}{2-v}}
		\esp\left[ \int_I dE
			\| \chi_x G_{\omega,L}(E) \chi_m\|^s
		\right]^{\frac{v}{s(2-v)}},
	\end{eqnarray}	
%	\begin{eqnarray}
%		\esp\left[ Q_{\omega,\Lambda}(x,m;I)\right]
%		\leq 
%		C (\lambda a_m)^{\frac{v-1}{2-v}} \| \rho \|_{\infty}^{\frac{1}{2-v}}
%		\esp\left[ \int_I dE
%			\| \chi_x (H^{\xi}_{\omega,m}-E)^{-1} \chi_m\|^s
%		\right]^{\frac{v}{s(2-v)}},
%	\end{eqnarray}			
	for all $m\in \Z$, $x\in\R$ and $L>0$.
\end{theorem}
%%%%%%%%%%%%%%%%%%%%%%%%%%%%%%%%%%%%%%%%%%%%%%%%%%%%%%%%%%%%
%%%%%%%%%%%%%%%%%%%%%%%%%%%%%%%%%%%%%%%%%%%%%%%%%%%%%%%%%%%%
The proof, given at the end of the section, will use several reduction steps discussed below.
%%%%%%%%%%%%%%%%%%%%%%%%%%%%%%%%%%%%%%%%%%%%%%%%%%%%%%%%%%%%
%%%%%%%%%%%%%%%%%%%%%%%%%%%%%%%%%%%%%%%%%%%%%%%%%%%%%%%%%%%%
We will need to work with the fractional eigenfunction correlator for which we introduce a family of perturbations of the finite volume operator $H_{\omega,L}$ so that
\begin{eqnarray}
	H_{\omega,L}^{m,\xi}
	=
	H_{\omega,L} + \lambda a_m( \xi - \omega_m)u_m.
\end{eqnarray}
This corresponds to setting the value of $\omega_m$ to $\xi$.
By the general theory summarized in \cite[Appendix B]{AENSS}, we know that the eigenvalues $(E_n)_n$ of these operators and their corresponding normalized eigenfunctions $(\varphi_n)_n$ can be chosen analytically in the parameter $z= \lambda a_m( \xi - \omega_m)$. We will sometimes denote $E_n=E_n(z)$ and $\varphi_n(\cdot)=\varphi_n(z)(\cdot)$ to stress this dependence. We also call $\Gamma_n$ the inverse of the function $z\mapsto E_n(z)$ which is shown to be well defined \cite{AENSS}. Note that $|\Gamma_n(E)| \leq 2 \lambda a_m M$ where $M=\max\{|\omega^-|,|\omega^+|\}$.

%%%%%%%%%%%%%%%%%%%%%%%%%%%%%%%%%%%%%%%%%%%%%%%%%%%%%%%%%%%%
%%%%%%%%%%%%%%%%%%%%%%%%%%%%%%%%%%%%%%%%%%%%%%%%%%%%%%%%%%%%
For $v\in[0,2]$, we define the $v$-fractional eigenfunction correlator as
\begin{eqnarray}\label{eq:eigenfunction-correlator}
	Q_{\omega,L}(x,m;I,v)
	=
	\sum_{n: E_n \in I}
	\bra \chi_x \varphi_n, \varphi_n \ket^{v/2}
	\bra u_m \varphi_n, \varphi_n \ket^{1-v/2},
\end{eqnarray}
where $(E_n)_n$ and $(\varphi_n)_n$ are chosen with the conventions above.
%%%%%%%%%%%%%%%%%%%%%%%%%%%%%%%%%%%%%%%%%%%%%%%%%%%%%%%%%%%%
%%%%%%%%%%%%%%%%%%%%%%%%%%%%%%%%%%%%%%%%%%%%%%%%%%%%%%%%%%%%
Note that
\begin{eqnarray}
	Q_{\omega,L}(x,m;I,0)
	&=&
	\text{Tr}  \left(  u_m P_I(H_{\omega,L}) \right),\label{v=0}
	\\
	Q_{\omega,L}(x,m;I,2)
	&=&
	\text{Tr}  \left( \chi_x P_I(H_{\omega,L}) \right).\label{v=2}
\end{eqnarray}
%%%%%%%%%%%%%%%%%%%%%%%%%%%%%%%%%%%%%%%%%%%%%%%%%%%%%%%%%%%%
%%%%%%%%%%%%%%%%%%%%%%%%%%%%%%%%%%%%%%%%%%%%%%%%%%%%%%%%%%%%
The next lemma allows us to control the correlator in \eqref{eq:correlator-box} through the fractional correlator \eqref{eq:eigenfunction-correlator} with $v=1$. We use the arguments of \cite{HSS} to by-pass a certain covering condition imposed on the single-site potential in \cite{AENSS}.
%%%%%%%%%%%%%%%%%%%%%%%%%%%%%%%%%%%%%%%%%%%%%%%%%%%%%%%%%%%%
%%%%%%%%%%%%%%%%%%%%%%%%%%%%%%%%%%%%%%%%%%%%%%%%%%%%%%%%%%%%
\begin{lemma}\label{thm:from-correlator-to-1-correlator}
	There exists a finite constant $c$ such that
	\begin{eqnarray}
		Q_{\omega,L}(x,m;I)
		\leq
		c\
		Q_{\omega,L}(x,m;I,1),
	\end{eqnarray}
	for all $x\in\R$ and $m\in\Z$.
\end{lemma}
%%%%%%%%%%%%%%%%%%%%%%%%%%%%%%%%%%%%%%%%%%%%%%%%%%%%%%%%%%%%
%%%%%%%%%%%%%%%%%%%%%%%%%%%%%%%%%%%%%%%%%%%%%%%%%%%%%%%%%%%%
\begin{proof}
	Note that
	\begin{eqnarray*}
		Q_{\omega,L}(x,m;I)
		&\leq&
		\sum_{n: E_n \in I} \| \chi_x P_{\varphi_n} \chi_m \|
		=
		\sum_{n: E_n \in I} \| \chi_x \varphi_n\| \, \| \varphi_n \chi_m \|,
	\end{eqnarray*}
	where $P_{\varphi_n}:=|\varphi_n\ket\bra\varphi_n|$ denotes the projector on the subspace spanned by $\varphi_n$.
	It follows from the hypothesis on the single-site potential together with Lemmas \ref{thm:general-bound-1} and \ref{thm:general-bound-2} that
	\begin{eqnarray*}
		\| u_m^{1/2} \varphi_n \|^2 
		\geq 
		c_u^2 \| \chi_{J+m} \varphi_n \|^2
		\geq
		c_1 \left( |\varphi_n(m)|^2+|\varphi'_n(m)|^2\right)
		\geq
		c_2
		\| \chi_m \varphi_n \|^2,
	\end{eqnarray*}
	for some constants $c_1$ and $c_2$ which are bounded away from $0$, uniformly in $m$.
	Hence,
	\begin{eqnarray*}
		Q_{\omega,L}(x,m;I)
		&\leq&
		c
		\sum_{n: E_n \in I} \| \chi_x \varphi_n\| \, \| u_m^{1/2} \varphi_n \|
		=
		c\
		Q_{\omega,L}(x,m;I,1),
	\end{eqnarray*}
	for some finite $c$ and for all $m$ and $x$.
\end{proof}
%%%%%%%%%%%%%%%%%%%%%%%%%%%%%%%%%%%%%%%%%%%%%%%%%%%%%%%%%%%%
%%%%%%%%%%%%%%%%%%%%%%%%%%%%%%%%%%%%%%%%%%%%%%%%%%%%%%%%%%%%
%With some considerations on the support of $U$ and estimates on solutions, we can prove that $\| \chi_m \varphi_n \| \leq C \| U_m^{1/2} \varphi_n \|$. Hence,
%\begin{eqnarray}
%	Q_{\omega,\Lambda}(x,m;I)
%	&\leq&
%	\sum_{n: E_n \in I} \| \chi_x P_{\varphi_n} \chi_m \|
%	=
%	\sum_{n: E_n \in I} \| \chi_x \varphi_n\| \, \| \varphi_n \chi_m \|
%	\\
%	&\leq&
%	C
%	\sum_{n: E_n \in I} \| \chi_x \varphi_n\| \, \| U_m^{1/2} \varphi_n \|
%	\\
%	&=&
%	C
%	Q_{\omega,\Lambda}(x,m;I,1).
%\end{eqnarray}
%%%%%%%%%%%%%%%%%%%%%%%%%%%%%%%%%%%%%%%%%%%%%%%%%%%%%%%%%%%%
%%%%%%%%%%%%%%%%%%%%%%%%%%%%%%%%%%%%%%%%%%%%%%%%%%%%%%%%%%%%
By the interpolation bound \cite[Lemma 2.1]{AENSS}, we have
\begin{eqnarray}\label{eq:interpolation}
	\esp \left[
		Q_{\omega,L}(x,m;I,1)
	\right]
	\leq
	\esp \left[
		Q_{\omega,L}(x,m;I,v)
	\right]^{1/(2-v)}
	\esp \left[
		Q_{\omega,L}(x,m;I,2)
	\right]^{(1-v)/(2-v)}.
\end{eqnarray}
%From \cite{AENSS}, equation (2.8), we also know that $Q_{\omega,\Lambda}(x,m;I,2)$ is bounded for each bounded $I\subset \R$, uniformly in $m$ and $\omega$. \gnote{MORE DETAILS?}
%
We notice that the correlator $Q_{\omega,L}(x,m;I,2)$ in \eqref{v=2} can be bounded deterministically. Indeed, if $I \subset (-\infty,E]$ then we have
\begin{eqnarray}
	Q_{\omega,L}(x,m;I,2)
	&\leq&
	\text{Tr} \, \left( \chi_x P_{(-\infty,E]}(H_{\omega,L}) \right)
	\\
	\label{eq:uniform-bound-Q2}
	&\leq&
	(|E|+B)^p \, \text{Tr} \left( \chi_x (H_{\omega,L}+B)^{-p} \right),
\end{eqnarray}
which is finite for $B$ and $p$ large enough, uniformly in $x$.
%%%%%%%%%%%%%%%%%%%%%%%%%%%%%%%%%%%%%%%%%%%%%%%%%%%%%%%%%%%%
%%%%%%%%%%%%%%%%%%%%%%%%%%%%%%%%%%%%%%%%%%%%%%%%%%%%%%%%%%%%

%At this point, we need to clarify the resampling. I think we have to take
%\begin{eqnarray}
%	H_{\omega,m}^{\xi}
%	=
%	H_{\omega} + \lambda a_m( \xi - \omega_m)U_m.
%\end{eqnarray}
%We can then define the eigenvalues and eigenfunctions analytically as $E_n=E_n(\xi)$ and $\psi_n(\cdot)=\phi_n(\xi)(\cdot)$, and let $\Gamma_n(E)$  be the inverse function of $E_n$. Then,
%%%%%%%%%%%%%%%%%%%%%%%%%%%%%%%%%%%%%%%%%%%%%%%%%%%%%%%%%%%%
%%%%%%%%%%%%%%%%%%%%%%%%%%%%%%%%%%%%%%%%%%%%%%%%%%%%%%%%%%%%
The following is the key identity to relate the correlator to the fractional moments of the resolvent. 
%%%%%%%%%%%%%%%%%%%%%%%%%%%%%%%%%%%%%%%%%%%%%%%%%%%%%%%%%%%%
%%%%%%%%%%%%%%%%%%%%%%%%%%%%%%%%%%%%%%%%%%%%%%%%%%%%%%%%%%%%
\begin{theorem}[\cite{AENSS}, Thm 2.1]
		If $\xi \neq \omega_m$ and $\sigma(H_{\omega,L}) \cap \sigma(H^{m,\xi}_{\omega,L}) \cap I = \emptyset$, then
		\begin{eqnarray}
			Q_{\omega,L}(x,m;I,v)
			&=&
			\sum_n \int_I dE \,
			\delta(\Gamma_n(E)+\lambda a_m(\xi-\omega_m)) |\Gamma_n(E)|^v
			\\
			\label{eq:identity-eigenfunction-correlator}
			&&
			\times 
			\| \chi_x (H^{m,\xi}_{\omega,L}-E)^{-1} u_m^{1/2}\psi_n(E) \|^v \, \| \psi_n(E) \|^{-v},
		\end{eqnarray}
		where
		\begin{eqnarray}
			\psi_n(E)(\cdot) = u_m^{1/2} \varphi_n(\Gamma_n(E))(\cdot).
		\end{eqnarray}
		Furthermore, for any $E$ and $a<b$ such that $E$ is not an eigenvalue of $H^{m,a}_{\omega,L}$ or $H^{m,b}_{\omega,L}$,
		\begin{eqnarray}
		\label{eq:delta-to-spectral-shift}
			\int^b_a d\omega_m \sum_n \delta(\Gamma_n(E)+\lambda a_m(\xi-\omega_m))
			=
			\lambda^{-1}a_m^{-1}
			\left(
				\text{Tr}P_E(H^{m,a}_{\omega,L})-\text{Tr}P_E(H^{m,b}_{\omega,L})
			\right),
		\end{eqnarray}
		where $P_E$ is the projection on $(-\infty,E]$.
\end{theorem}
%%%%%%%%%%%%%%%%%%%%%%%%%%%%%%%%%%%%%%%%%%%%%%%%%%%%%%%%%%%%
%%%%%%%%%%%%%%%%%%%%%%%%%%%%%%%%%%%%%%%%%%%%%%%%%%%%%%%%%%%%
Note that, as we consider absolutely continuous environments, the hypothesis of the theorem are almost surely satisfied for each choice of the parameters (see for instance \cite[Lemma B.2]{AENSS}).
%%%%%%%%%%%%%%%%%%%%%%%%%%%%%%%%%%%%%%%%%%%%%%%%%%%%%%%%%%%%
%%%%%%%%%%%%%%%%%%%%%%%%%%%%%%%%%%%%%%%%%%%%%%%%%%%%%%%%%%%%
If we take $a=\omega^-$ and $b=\omega^+$ in the right-hand-side of \eqref{eq:delta-to-spectral-shift}, we recover the spectral shift
\begin{eqnarray}
	S_{\omega,m}(L,E) = \text{Tr}P_E(H^{m,\omega^+}_{\omega,L})-\text{Tr}P_E(H^{m,\omega^-}_{\omega,L}).
\end{eqnarray}
By a combination \cite[Theorem 2.1 and Proposition 5.1]{CHN}, we know that $S_{\omega,m}(L,E)$ has finite moments of order $p\geq 1$ with respect to the Lebesgue measure, uniformly in $\omega$ and $m$. The uniformity in $m$ can be seen from  \cite[Formula (5.6)]{CHN} where the dependence on the single-site potential is made explicit. We state this as a lemma.
%%%%%%%%%%%%%%%%%%%%%%%%%%%%%%%%%%%%%%%%%%%%%%%%%%%%%%%%%%%%
%%%%%%%%%%%%%%%%%%%%%%%%%%%%%%%%%%%%%%%%%%%%%%%%%%%%%%%%%%%%
\begin{lemma}\label{thm:p-moments-spectral-shift}
	For all $p\geq 1$ and every compact interval $I \subset (0,\infty)$,
	there exists a constant $C=C(p,I)$ such that
	\begin{eqnarray*}
		\int_I dE \left| S_{\omega,m}(L,E) \right|^p \leq C,
	\end{eqnarray*}
	for all $L>0$, $m\in\Z$ and $\omega\in\Omega$.
\end{lemma}
%%%%%%%%%%%%%%%%%%%%%%%%%%%%%%%%%%%%%%%%%%%%%%%%%%%%%%%%%%%%
%%%%%%%%%%%%%%%%%%%%%%%%%%%%%%%%%%%%%%%%%%%%%%%%%%%%%%%%%%%%
The next lemma brings the relation between the fractional eigenfunction correlators and the fractional moments of the Green's function. This finishes the proof of Theorem \ref{thm:from-green-to-correlator}.
%%%%%%%%%%%%%%%%%%%%%%%%%%%%%%%%%%%%%%%%%%%%%%%%%%%%%%%%%%%%
%%%%%%%%%%%%%%%%%%%%%%%%%%%%%%%%%%%%%%%%%%%%%%%%%%%%%%%%%%%%
\begin{lemma}\label{thm:from-green-to-eigenfunction-correlator}
	For all $v\in(0,1)$, $s\in(v,1)$ and all compact interval $I\subset (0,\infty)$, there exists a constant $C=C(v,s,I)$ such that
%	\begin{eqnarray}
%		\esp_m\left[
%			Q_{\omega,\Lambda}(x,m;I,v)
%		\right]
%		&\leq&
%		C (\lambda a_m)^{v-1}\| \rho \|_{\infty} \,
%		\esp\left[ \int_I dE
%			\| \chi_x (H_{\omega,\Lambda}-E)^{-1} \chi_m\|^s
%		\right]^{v/s},
%	\end{eqnarray}
\begin{eqnarray*}
		\esp\left[
			Q_{\omega,L}(x,m;I,v)
		\right]
		&\leq&
		C (\lambda a_m)^{v-1}\| \rho \|_{\infty} \,
		\esp\left[ \int_I dE
			\| \chi_x G_{\omega,L}(E) \chi_m\|^s
		\right]^{v/s},
	\end{eqnarray*}
	for all $x\in\R$, $m\in\Z$ and $L>0$.
\end{lemma}
%%%%%%%%%%%%%%%%%%%%%%%%%%%%%%%%%%%%%%%%%%%%%%%%%%%%%%%%%%%%
%%%%%%%%%%%%%%%%%%%%%%%%%%%%%%%%%%%%%%%%%%%%%%%%%%%%%%%%%%%%
\begin{proof}
	Let us denote by $\esp_m$ the expected value with respect to $\omega_m$, which corresponds to integration against $\rho(\omega_m)d\omega_m$ over the interval $[\omega^-,\omega^+]$. Remember that $H^{m,\xi}_{\omega,L}$ is independent of $\omega_m$. Averaging \eqref{eq:identity-eigenfunction-correlator} with respect to $\esp_m$, recalling that $|\Gamma_n(E)| \leq 2 \lambda a_m M$ with $M=\max\{|\omega^-|,|\omega^+|\}$ and using \eqref{eq:delta-to-spectral-shift},
	\begin{eqnarray}\label{representation 1}
		\esp_m\left[
			Q_{\omega,L}(x,m;I,v)
		\right]
		&\leq&
		(2 \lambda a_m M)^v \int_I dE \, 
		\left(
		\int d\omega_m \rho(\omega_m)
			\sum_n \delta(\Gamma_n(E)+\lambda a_m(\xi-\omega_m))
		\right)\notag
		\\
		&&
		\phantom{blablablablablabla}
		\times
		\| \chi_x (H^{m,\xi}_{\omega,L}-E)^{-1} u_m^{1/2}\|^v
		\\
		&\leq& 
		2^v M^v(\lambda a_m)^{v-1} \| \rho \|_{\infty}\int_I dE \, 
		S_m(L,E)
		\| \chi_x (H^{m,\xi}_{\omega,L}-E)^{-1} u_m^{1/2}\|^v.\notag
	\end{eqnarray}
	Integrating the inequality \eqref{representation 1} with respect to $\rho(\xi)d\xi$ and then with respect to $\{\omega_n:\, n\neq m\}$, we obtain
	\begin{eqnarray}\label{representation 2}
		%\nonumber
		\esp\left[
			Q_{\omega,L}(x,m;I,v)
		\right]
		&\leq&
		2^v M^v(\lambda a_m)^{v-1} \| \rho \|_{\infty}
		\esp\left[
		\int_I dE \, 
		S_m(L,E)
		\| \chi_x G_{\omega,L}(E) u_m^{1/2}\|^v
		\right].
	\end{eqnarray}
	Next, we apply H\"older's inequality to \eqref{representation 2} with respect to $\p \times dE$ to get
	\begin{eqnarray*}
		\esp\left[
			Q_{\omega,L}(x,m;I,v)
		\right]
		&\leq&
		2^v M^v(\lambda a_m)^{v-1}\| \rho \|_{\infty} \,
		\esp\left[
			\int_I dE S_m(L,E)^{s/(s-v)}
		\right]^{(s-v)/s}
		\\
		&&
		\phantom{blablablabla}
		\times
		\esp\left[ \int_I dE
			\| \chi_x G_{\omega,L}(E) u_m^{1/2}\|^s
		\right]^{v/s}
		\\
		&&
		\leq
		C (\lambda a_m)^{v-1}\| \rho \|_{\infty} \,
		\esp\left[ \int_I dE
			\| \chi_x G_{\omega,L}(E) \chi_m\|^s
		\right]^{v/s},
	\end{eqnarray*}
	where we used Lemma \ref{thm:p-moments-spectral-shift} in the last step.
\end{proof}
%%%%%%%%%%%%%%%%%%%%%%%%%%%%%%%%%%%%%%%%%%%%%%%%%%%%%%%%%%%%
%%%%%%%%%%%%%%%%%%%%%%%%%%%%%%%%%%%%%%%%%%%%%%%%%%%%%%%%%%%%
We finish the proof of Theorem \ref{thm:from-green-to-correlator}:
%%%%%%%%%%%%%%%%%%%%%%%%%%%%%%%%%%%%%%%%%%%%%%%%%%%%%%%%%%%%
%%%%%%%%%%%%%%%%%%%%%%%%%%%%%%%%%%%%%%%%%%%%%%%%%%%%%%%%%%%%
\begin{proof}[Proof of Theorem \ref{thm:from-green-to-correlator}]
	The result follows from Lemma \ref{thm:from-correlator-to-1-correlator} together with the interpolation bound \eqref{eq:interpolation} and the uniform bound \eqref{eq:uniform-bound-Q2} on $Q_{\omega,L}(x,m;I,2)$, and Lemma \ref{thm:from-green-to-eigenfunction-correlator}.
\end{proof}
%%%%%%%%%%%%%%%%%%%%%%%%%%%%%%%%%%%%%%%%%%%%%%%%%%%%%%%%%%%%
%%%%%%%%%%%%%%%%%%%%%%%%%%%%%%%%%%%%%%%%%%%%%%%%%%%%%%%%%%%%
Finally, we give the proof of our main result Theorem \ref{thm:DL}.
%%%%%%%%%%%%%%%%%%%%%%%%%%%%%%%%%%%%%%%%%%%%%%%%%%%%%%%%%%%%
%%%%%%%%%%%%%%%%%%%%%%%%%%%%%%%%%%%%%%%%%%%%%%%%%%%%%%%%%%%%
\begin{proof}[Proof of Theorem \ref{thm:DL}]
	Since $H_{\omega,L}$ converges to $\opH$ in the strong resolvent sense as $L\to\infty$, we have
	\begin{eqnarray*}
		Q_{\omega,\lambda}(x,y;I)
		\leq
		\liminf_L
		Q_{\omega,L}(x,y;I).
	\end{eqnarray*}
	By Fatou's lemma and using that $Q_{\omega,\lambda}(m,n;I)\leq 1$, we have
	\begin{eqnarray*}
		\esp\left[ 
			Q_{\omega,\lambda}(x,y;I)^2
		\right]
		\leq
		\liminf_L
		\esp\left[
			Q_{\omega,L}(x,y;I)
		\right].
	\end{eqnarray*}
	The result follows from Theorem \ref{thm:from-green-to-correlator} and the uniform bound of Theorem \ref{thm:FM}.
\end{proof}
%%%%%%%%%%%%%%%%%%%%%%%%%%%%%%%%%%%%%%%%%%%%%%%%%%%%%%%%%%%%
%%%%%%%%%%%%%%%%%%%%%%%%%%%%%%%%%%%%%%%%%%%%%%%%%%%%%%%%%%%%

%%%%%%%%%%%%%%%%%%%%%%%%%%%%%%%%%%%%%%%%%%%%%%%%%%%%%%%%%%%%
%%%%%%%%%%%%%%%%%%%%%%%%%%%%%%%%%%%%%%%%%%%%%%%%%%%%%%%%%%%%
%%%%%%%%%%%%%%%%%%%%%%%%%%%%%%%%%%%%%%%%%%%%%%%%%%%%%%%%%%%%
%%%%%%%%%%%%%%%%%%%%%%%%%%%%%%%%%%%%%%%%%%%%%%%%%%%%%%%%%%%%
%%%%%%%%%%%%%%%%%%%%%%%%%%%%%%%%%%%%%%%%%%%%%%%%%%%%%%%%%%%%
%%%%%%%%%%%%%%%%%%%%%%%%%%%%%%%%%%%%%%%%%%%%%%%%%%%%%%%%%%%%
%%%%%%%%%%%%%%%%%%%%%%%%%%%%%%%%%%%%%%%%%%%%%%%%%%%%%%%%%%%%
%%%%%%%%%%%%%%%%%%%%%%%%%%%%%%%%%%%%%%%%%%%%%%%%%%%%%%%%%%%%

\section{Proof of Theorem \ref{thm:decay-eigenfunctions} and \ref{thm:lower-bound-DL}}\label{sec:consequences}

%%%%%%%%%%%%%%%%%%%%%%%%%%%%%%%%%%%%%%%%%%%%%%%%%%%%%%%%%%%%
%%%%%%%%%%%%%%%%%%%%%%%%%%%%%%%%%%%%%%%%%%%%%%%%%%%%%%%%%%%%

%%%%%%%%%%%%%%%%%%%%%%%%%%%%%%%%%%%%%%%%%%%%%%%%%%%%%%%%%%%%
%%%%%%%%%%%%%%%%%%%%%%%%%%%%%%%%%%%%%%%%%%%%%%%%%%%%%%%%%%%%
\begin{proof}[Proof of Theorem \ref{thm:decay-eigenfunctions}]
	We will establish the lower bound for $\sqrt{|\vp_{\omega,E}(x)|^2 + |\vp'_{\omega,E}(x)|^2}$ since the bound for $\| \chi_x \vp_{\omega,E} \|$ will then follow from Lemma \ref{thm:general-bound-2}.
	Recall that we can reconstruct 
	$\Psi_{\omega,E}
	=
	\begin{pmatrix}
		\vp \\ \vp'
	\end{pmatrix}
	$ 
	using the transfer matrices as
	$\displaystyle
		\Psi_{\omega,E}(x)			
	=
	\T(x,0;E)
	\psi_0
	$
	for some possibly random $\| \psi_0\|=1$.
	This implies in particular that
	\begin{eqnarray*}
		\| \Psi_{\omega,E}(x)\|
		\geq 
		\| \T(x,0;E)\|^{-1}.
	\end{eqnarray*}	
	Assume $x>0$, the opposite case being analogous.
	Using Lemma \ref{thm:comparison} with some $\vartheta_1 \neq \vartheta_2$, 
	\begin{eqnarray*}
		\nonumber
		\p\left[ \| \Psi_{\omega,E}(x)\|  \leq \e^{-c_1 |x|^{1-2\alpha}}\right]
		&\leq&
		\p\left[ \| \T(x,0;E) \| \geq \e^{c_1 |x|^{1-2\alpha}}\right]\\
		&\leq& 
		\e^{-2c_1 |x|^{1-2\alpha}}\esp\left[ \|  \T(x,0;E)\|^2\right]\\
		&\leq& 
		C_1(\vartheta_1,\vartheta_2) \e^{-2c_1 |x|^{1-2\alpha}}
		\left\{
			 \esp\left[ R^2(x,\vartheta_1)\right]+\esp\left[ R^2(x,\vartheta_2)\right]
		\right\},
	\end{eqnarray*}
	for some $C_1(\vartheta_1,\vartheta_2) >0$. From the martingale decomposition \eqref{eq:martingale-decomposition}, one has
	\begin{eqnarray}
		R(x,\vartheta_1)
		=
		\prod^{\lfloor x \rfloor}_{j=1}
			\exp\left\{
				\frac{\omega_j}{j^{\alpha}}A_j + \frac{\omega_j^2}{j^{2\alpha}}B_j + E_j)
			\right\},
	\end{eqnarray}
	where $A_j$ and $B_j$ are independent of $\omega_j$ and bounded uniformly in $E\in I$, and $E_j=o(j^{-2\alpha})$, uniformly in  $E\in I$. Hence, from standard estimates on the exponential moments of bounded centered random variables,
	\begin{eqnarray*}
		\esp[R(x,\vartheta_1)^2]
		\leq
		\prod^{\lfloor x \rfloor}_{j=1}
		\left(
			1 + \frac{C_2}{j^{2\alpha}} \esp[\omega_j^2] + o(j^{-2\alpha})
		\right)
		\leq
		\e^{C_3x^{1-2\alpha}},
	\end{eqnarray*}
	for some finite constants $C_2=C_2(I)$ and $C_3=C_3(I)$. The bound for $R^2(x,\vartheta_2)$ is of course similar. The result follows by Borel-Cantelli choosing $2c_1>C_3$.
	
%	From \eqref{eq:ODE-prufer}, we obtain
%	\begin{eqnarray}
%		R(x,\vartheta_1)\leq \exp\left( C_1 \sum^{\lfloor x \rfloor +1}_{j=0} \frac{\omega_j}{j^{\alpha}} \right),
%	\end{eqnarray}
%	for some $C_1=C_1(I)$. By standard estimates on the exponential moments of bounded centered random variables,
%	\begin{eqnarray}
%		\esp\left[
%			\exp\left( C \sum^{\lfloor x \rfloor +1}_{j=0} \frac{\omega_j}{j^{\alpha}} \right)
%		\right]
%		=
%		\prod^{\lfloor x \rfloor +1}_{j=0}
%		\left(
%			1 + \frac{C^2}{j^{2\alpha}} \esp[\omega_j^2] + o(j^{-2\alpha})
%		\right)
%		\leq
%		e^{C_2x^{1-2\alpha}},
%	\end{eqnarray}
%	for some $C_2>0$. 
	The upper bound is quite standard. Following for instance the proof of \cite[Theorem 9.22]{CFKS}, we obtain
	\begin{eqnarray*}
		\| \chi_x \vp_{\omega,E} \| \| \chi_0 \vp_{\omega,E} \| \leq C_{\omega} \e^{-c_2|x|^{1-2\alpha}},
	\end{eqnarray*}
	for some random almost surely finite $C_{\omega}>0$ and deterministic $c_2>0$. We can use the lower bound we just proved to get a lower bound on $\| \chi_0 \vp_{\omega,E} \|$ uniformly in $E\in I$. 
\end{proof}
%%%%%%%%%%%%%%%%%%%%%%%%%%%%%%%%%%%%%%%%%%%%%%%%%%%%%%%%%%%%
%%%%%%%%%%%%%%%%%%%%%%%%%%%%%%%%%%%%%%%%%%%%%%%%%%%%%%%%%%%%
\begin{proof}[Proof of Theorem \ref{thm:lower-bound-DL}]
	For the first statement, let $\kappa<1-2\alpha$. Then,
	\begin{eqnarray*}
		&&\esp\left[\sup_t
					\left\|
						\e^{ \frac12 |X|^{\kappa}}\e^{-it\opH} P_I(\opH)\chi_m
					\right\|^2
				\right]\\
		&& \phantom{blablablab} =
		\esp\left[\sup_t
					\langle
						e^{  |X|^{\kappa}}\e^{-it\opH}P_I(\opH)\chi_m,
						\e^{-it\opH}\chi_m
					\rangle
				\right]\\
		&& \phantom{blablablab} \leq
		\esp\left[\sup_t
					\sum_n
					\left|
					\langle
						\e^{  |X|^{\kappa}}\e^{-it\opH}P_I(\opH)\chi_m, \chi_n
					\rangle
					\right|
					\left|
					\langle
						\chi_n,
						\e^{-it\opH}P_I(\opH)\chi_m
					\rangle
					\right|
				\right]\\
		&& \phantom{blablablab} 
		\leq
		C\
		\esp\left[\sup_t
					\sum_n
					\e^{|n|^{\kappa}}
					\left|
					\langle
						\chi_n,
						\e^{-it\opH}P_I(\opH)\chi_m
					\rangle
					\right|^2
				\right]\\
		&& \phantom{blablablab} 
		=
		C\
		\esp\left[
					\sum_n
					\e^{|n|^{\kappa}}
					Q_{\omega,\lambda}(m,n;I)^2
				\right],
	\end{eqnarray*}
	for some finite $C>0$. The last sum is finite for each $m$ in virtue of \eqref{eq:dynamical-localization}.
	
	For the second statement, let $c_{\omega}$ and $c_1$ be as in Theorem \ref{thm:decay-eigenfunctions}. Let $(\psi_l)_{l}$ be a basis of ${\text Ran}P_I(\opH)$ given by normalized eigenfunctions of the operator $\opH$ with corresponding eigenvalues $(E_l)_l$.
	 Let $N\geq 1$. Taking $\kappa>1-2\alpha$ and applying the lower bound in Theorem \ref{thm:decay-eigenfunctions},
	 %applying Lemma \ref{thm:general-bound-2} and the lower bound in Proposition \ref{thm:decay-eigenfunctions},
	 \begin{eqnarray*}
	 	\| \e^{\frac12 |X|^{\kappa}} \chi_{[0,N]}\psi_l\|^2
	 	&=&
	 	\int^N_0 \overline{\psi}_l(x) \e^{|x|^{\kappa}}\psi_l(x)\, dx
	 	\geq
	 	\sum^{N-1}_{n=0}\e^{n^{\kappa}} \int^{n+1}_n |\psi_l(x)|^2dx
%	 	\\
%	 	&\geq&
%	 	C \sum^{N-1}_{n=0}e^{n^{\kappa}} \left( |\psi_l(n)|^2 + |\psi'(n)|^2\right)
	 	\\
	 	&\geq&
	 	C \sum^{N-1}_{n=0}\e^{n^{\kappa}}c_{\omega}\e^{-c_1 n^{1-2\alpha}}
	 	\geq
	 	c'_{\omega} \e^{\frac12 N^{\kappa}},
	 \end{eqnarray*}
	 for some $C>0$ and some suitable random quantity $c'_{\omega}>0$.
	Let $\psi \in {\text Ran}P_I(\opH)$ and write $\psi = \sum_l a_l \psi_l$ with $\sum_l |a_l|^2=1$. Then,
	\begin{eqnarray}
		\left\| \e^{\frac12 |X|^{\kappa}} \e^{-i t \opH}\psi\right\|^2
		=
		\sum_{l,l'} a_l \overline{a}_{l'}\e^{-it(E_l - E_{l'})}
		\langle \psi_{l'}, \e^{|X|^{\kappa}} \psi_l \rangle.
	\end{eqnarray}
	A careful application of the dominated convergence theorem to exchange sums and integrals yields
	\begin{eqnarray}
		\lim_{T\to\infty} \frac{1}{T} \int^T_0 \left\| \e^{\frac12 |X|^{\kappa}} \e^{-i t \opH}\psi\right\|^2 dt
		=
		\sum_l |a_l|^2 \left\| \e^{\frac12 |X|^{\kappa}} \psi_l\right\|^2
		\geq 
		c'_{\omega} \e^{\frac12 N^{\kappa}}.
	\end{eqnarray}
	Hence, there exists an diverging (random) sequence $(T_N)_N$ such that
	\begin{eqnarray}
		\frac{1}{T_N} \int^{T_N}_0 \left\| \e^{\frac12 |X|^{\kappa}}\e^{-i t \opH}\psi\right\|^2 dt
		\geq 
		\frac{c'_{\omega}}{2} \e^{\frac12 N^{\kappa}},	
		\end{eqnarray}
	for all $N\geq 1$. From here, we can find a diverging (random) sequence $(t_N)_N$ such that
	\begin{eqnarray}
		\left\| \e^{\frac12 |X|^{\kappa}} \e^{-i t_N \opH}\psi\right\|^2
		\geq 
		\frac{c'_{\omega}}{4} \e^{\frac12 N^{\kappa}},
	\end{eqnarray}
	for all $N\geq 1$. This finishes the proof.
\end{proof}
%%%%%%%%%%%%%%%%%%%%%%%%%%%%%%%%%%%%%%%%%%%%%%%%%%%%%%%%%%%%
%%%%%%%%%%%%%%%%%%%%%%%%%%%%%%%%%%%%%%%%%%%%%%%%%%%%%%%%%%%%
%%%%%%%%%%%%%%%%%%%%%%%%%%%%%%%%%%%%%%%%%%%%%%%%%%%%%%%%%%%%
%%%%%%%%%%%%%%%%%%%%%%%%%%%%%%%%%%%%%%%%%%%%%%%%%%%%%%%%%%%%
%%%%%%%%%%%%%%%%%%%%%%%%%%%%%%%%%%%%%%%%%%%%%%%%%%%%%%%%%%%%
%%%%%%%%%%%%%%%%%%%%%%%%%%%%%%%%%%%%%%%%%%%%%%%%%%%%%%%%%%%%

\appendix

%%%%%%%%%%%%%%%%%%%%%%%%%%%%%%%%%%%%%%%%%%%%%%%%%%%%%%%%%%%%
%%%%%%%%%%%%%%%%%%%%%%%%%%%%%%%%%%%%%%%%%%%%%%%%%%%%%%%%%%%%
%%%%%%%%%%%%%%%%%%%%%%%%%%%%%%%%%%%%%%%%%%%%%%%%%%%%%%%%%%%%
%%%%%%%%%%%%%%%%%%%%%%%%%%%%%%%%%%%%%%%%%%%%%%%%%%%%%%%%%%%%
%%%%%%%%%%%%%%%%%%%%%%%%%%%%%%%%%%%%%%%%%%%%%%%%%%%%%%%%%%%%
%%%%%%%%%%%%%%%%%%%%%%%%%%%%%%%%%%%%%%%%%%%%%%%%%%%%%%%%%%%%
%%%%%%%%%%%%%%%%%%%%%%%%%%%%%%%%%%%%%%%%%%%%%%%%%%%%%%%%%%%%
%%%%%%%%%%%%%%%%%%%%%%%%%%%%%%%%%%%%%%%%%%%%%%%%%%%%%%%%%%%%

\section{General estimates}\label{app:general}

We quote two lemmas from \cite{HSS} that were used repeatedly in the proofs.
%%%%%%%%%%%%%%%%%%%%%%%%%%%%%%%%%%%%%%%%%%%%%%%%%%%%%%%%%%%%
%%%%%%%%%%%%%%%%%%%%%%%%%%%%%%%%%%%%%%%%%%%%%%%%%%%%%%%%%%%%
\begin{lemma}[\cite{HSS}, Lemma A.1]
	\label{thm:general-bound-1}
	For every $q \in L^1_{loc}(\R)$, every $c<d$ and every solution of $-\Delta \vp + q\vp = 0$ on $[c,d]$, we have
	\begin{eqnarray*}
		&&
		\left( |\vp(c)|^2 + |\vp'(c)|^2 \right)
		\exp\left( -\int^d_c |1+q(x)|\, dx\right)
		\leq
		\left( |\vp(d)|^2 + |\vp'(d)|^2 \right)
		\\
		&&
		\phantom{blablablablablablablablablablablablablabla}
		\leq
		\left( |\vp(c)|^2 + |\vp'(c)|^2 \right)
		\exp\left( \int^d_c |1+q(x)|\, dx\right).
	\end{eqnarray*}
\end{lemma}
%%%%%%%%%%%%%%%%%%%%%%%%%%%%%%%%%%%%%%%%%%%%%%%%%%%%%%%%%%%%
%%%%%%%%%%%%%%%%%%%%%%%%%%%%%%%%%%%%%%%%%%%%%%%%%%%%%%%%%%%%

%%%%%%%%%%%%%%%%%%%%%%%%%%%%%%%%%%%%%%%%%%%%%%%%%%%%%%%%%%%%
%%%%%%%%%%%%%%%%%%%%%%%%%%%%%%%%%%%%%%%%%%%%%%%%%%%%%%%%%%%%
\begin{lemma}[\cite{HSS}, Lemma A.2]
	\label{thm:general-bound-2}
	For any positive numbers $l$ and $M$, there exists $C>0$ such that
	\begin{eqnarray}
		\int^{c+l}_c |\vp(t)|^2 dt \geq C \left( |\vp(c)|^2 + |\vp'(c)|^2 \right),
	\end{eqnarray}
	for all $c\in \R$, all $q \in L^1_{loc}(\R)$ such that $\int^{c+l}_c |1+q(x)|\, dx \leq M$ and every solution of $-\Delta \vp + q\vp = 0$ on $[c,c+l]$.
\end{lemma}
%%%%%%%%%%%%%%%%%%%%%%%%%%%%%%%%%%%%%%%%%%%%%%%%%%%%%%%%%%%%
%%%%%%%%%%%%%%%%%%%%%%%%%%%%%%%%%%%%%%%%%%%%%%%%%%%%%%%%%%%%
%%%%%%%%%%%%%%%%%%%%%%%%%%%%%%%%%%%%%%%%%%%%%%%%%%%%%%%%%%%%
%%%%%%%%%%%%%%%%%%%%%%%%%%%%%%%%%%%%%%%%%%%%%%%%%%%%%%%%%%%%
%%%%%%%%%%%%%%%%%%%%%%%%%%%%%%%%%%%%%%%%%%%%%%%%%%%%%%%%%%%%
%%%%%%%%%%%%%%%%%%%%%%%%%%%%%%%%%%%%%%%%%%%%%%%%%%%%%%%%%%%%

\section{The martingale decomposition}\label{app:martingale}
We briefly describe the martingale analysis of \cite[Section 9]{KLS} needed in Lemma \ref{thm:bounds-on-Tmn}. Note that the analysis of \cite{KLS} is almost sure. We only require a version in expectation.
\newline
Let $I\subset \R$ be a compact interval and, for $E\in I$, set $k=\sqrt{E}$. Letting $\bar{\theta}_n(y)=\theta(n)+ky$, it can be showed that
\begin{eqnarray}
	\log \frac{R(n)}{R(m)}
	&=&
	\sum^{n-1}_{j=m}
	\frac{\lambda\omega_j}{2kj^{\alpha}}
	\int^1_0 u(y) 
	\sin \left( 2 \bar{\theta}_j(y)\right)\, dy
	\notag
	\\
	&&
	-
	\sum^{n-1}_{j=m}
	\frac{\lambda^2\omega_j^2}{4k^2j^{2\alpha}}
	\int^1_0 u(y) 
	\left( \int^y_0 u(t) \, dt \right)
	\cos 2 \bar{\theta}_j(y)
	dy
	\notag
	\\
	&&
	+
	\sum^{n-1}_{j=m}
	\frac{\lambda^2\omega_j^2}{8k^2j^{2\alpha}}
	\left|
	\int^1_0 u(y) 
	\e^{2iky}\, dy
	\right|^2 \cos(4(\theta(j)-\nu_k))
	\\
	\notag
	&&
	\label{eq:martingale-decomposition}
	+
	\sum^{n-1}_{j=m}
	\frac{\lambda^2\omega_j^2}{8k^2j^{2\alpha}}
	\left|
	\int^1_0 u(y) 
	\e^{2iky}\, dy
	\right|^2
	+
	K_{\omega}(m,n),
\end{eqnarray}
for some $\nu_k\in[0,2\pi)$ and where
\begin{eqnarray*}
	| K_{\omega}(m,n)|
	=
	o\left( \sum^{n-1}_{j=m} j^{-2\alpha}\right),
\end{eqnarray*}
uniformly in $E\in I$ and $\omega$ (the bound is indeed deterministic). From \eqref{eq:ODE-prufer}, we can see that $\{ \bar{\theta}_j(y): y\geq 0\}$ and $\omega_j$ are independent. Recalling that the $\omega_j$'s are centered and satisfy $\esp[\omega_j^2]=1$, we can integrate \eqref{eq:martingale-decomposition},
\begin{eqnarray*}
	\esp\left[
		\log \frac{R(n)}{R(m)}
	\right]
	&=&
	-
	\frac{\lambda^2}{4k^2}
	\esp\left[
		\sum^{n-1}_{j=m}
		\frac{1}{j^{2\alpha}}
		\int^1_0 u(y) 
		\left( \int^y_0 u(t) \, dt \right)
		\cos 2 \bar{\theta}_j(y)
		dy
	\right]
	\\
	&&
	+
	\frac{\lambda^2}{8k^2}
	\esp\left[
		\sum^{n-1}_{j=m}
		\frac{1}{j^{2\alpha}}
		\left|
		\int^1_0 u(y) 
		\e^{2iky}\, dy
		\right|^2 \cos(4(\theta(j)-\nu_k))
	\right]
	\\
	&&
	+
	\frac{\lambda^2}{8k^2}
	\sum^{n-1}_{j=m}
	\frac{1}{j^{2\alpha}}
	\left|
	\int^1_0 u(y) 
	\e^{2iky}\, dy
	\right|^2
	+
	\esp\left[
		E_{\omega}(m,n)
	\right].
\end{eqnarray*}
The deterministic analysis of \cite[Section 9]{KLS} allows us to control the first two sums on the right-hand-side to show that they are $ o\left( \sum^{n-1}_{j=m} j^{-2\alpha}\right)$, uniformly in $E\in I$ such that $\sqrt{E}\notin \pi \Z$. This shows \eqref{eq:asymptotics-prufer}.

%%%%%%%%%%%%%%%%%%%%%%%%%%%%%%%%%%%%%%%%%%%%%%%%%%%%%%%%%%%%
%%%%%%%%%%%%%%%%%%%%%%%%%%%%%%%%%%%%%%%%%%%%%%%%%%%%%%%%%%%%
%%%%%%%%%%%%%%%%%%%%%%%%%%%%%%%%%%%%%%%%%%%%%%%%%%%%%%%%%%%%
%%%%%%%%%%%%%%%%%%%%%%%%%%%%%%%%%%%%%%%%%%%%%%%%%%%%%%%%%%%%
%%%%%%%%%%%%%%%%%%%%%%%%%%%%%%%%%%%%%%%%%%%%%%%%%%%%%%%%%%%%
%%%%%%%%%%%%%%%%%%%%%%%%%%%%%%%%%%%%%%%%%%%%%%%%%%%%%%%%%%%%

\section{Pure point spectrum}\label{app:pp}

%%%%%%%%%%%%%%%%%%%%%%%%%%%%%%%%%%%%%%%%%%%%%%%%%%%%%%%%%%%%
%%%%%%%%%%%%%%%%%%%%%%%%%%%%%%%%%%%%%%%%%%%%%%%%%%%%%%%%%%%%
\begin{proposition}
	Assume dynamical localization for $\opH$ holds in the sense of \eqref{eq:dynamical-localization} in an energy interval $I\subset\R$. Then, the spectrum of $\opH$ is almost surely pure point in $I$.
\end{proposition}
%%%%%%%%%%%%%%%%%%%%%%%%%%%%%%%%%%%%%%%%%%%%%%%%%%%%%%%%%%%%
%%%%%%%%%%%%%%%%%%%%%%%%%%%%%%%%%%%%%%%%%%%%%%%%%%%%%%%%%%%%
\begin{proof}
	This is a consequence of the RAGE Theorem \cite{CFKS}. Suppose that \eqref{eq:dynamical-localization} holds in an energy interval $I$ and consider $\chi_R$, the projector on the box $[-R,R]$. As $\chi_R$ converges strongly to the identity, it is enough to show that, $\p$-almost surely,
	\begin{eqnarray}\label{eq:RAGE}
		\lim_{R\to\infty} \sup_t \left\| \left( 1 - \chi_R \right) \e^{-it\opH}P_I(\opH) \chi_m \right\|^2 =0,
	\end{eqnarray}
	for all integer $m$ as this implies that the range of $P_I(\opH)$ is almost surely included in the point spectrum of $\opH$. Now,
	\begin{eqnarray*}
		\left\| \left( 1 - \chi_R \right) \e^{-it\opH}P_I(\opH) \chi_m \right\|^2 
		&\leq&
		\sum_{|n|>R} \left| \langle \chi_n, \e^{-it\opH}P_I(\opH) \chi_m \rangle\right|^2\\
%		&=&
%		\sum_{|m|>R} \left| \langle e^{itH_{\omega}} \delta_m, P_I(H_{\omega}) \delta_n \rangle\right|^2\\
%		&=&
%		\sum_{|m|>R} \left| \langle e^{itH_{\omega}} \delta_m, P_I(H_{\omega})^2 \delta_n \rangle\right|^2\\
%		&=&
%		\sum_{|m|>R} \left| \langle P_I(H_{\omega}) e^{itH_{\omega}} \delta_m, P_I(H_{\omega}) \delta_n \rangle\right|^2\\
		&=&
		\sum_{|n|>R} \left| \langle P_I(\opH) \ \e^{it\opH} \chi_n, \chi_m \rangle\right|^2\\
		&\leq&
		\sum_{|n|>R} Q_\omega(m,n;I)^2,
	\end{eqnarray*}
	for all $t\in\R$. By Fatou's lemma,
	% and using that $Q_{\omega,\lambda}(m,n;I)\leq 1$,
	\begin{eqnarray*}
		\esp\left[  \lim_{R\to\infty} \sup_t \left\| \left( 1 - \chi_R \right) \e^{-it\opH}P_I(\opH) \chi_m \right\|^2 \right]
		&\leq&
		\esp\left[  \liminf_{R\to\infty} 
		\sum_{|n|>R} Q_{\omega,\lambda}(m,n;I)^2
		\right]\\
		&\le&
		\liminf_{R\to\infty} \esp\left[   
		\sum_{|n|>R} Q_{\omega,\lambda}(m,n;I)^2
		\right]
		=0,
	\end{eqnarray*}
	which shows that \eqref{eq:RAGE} holds $\p$-almost surely for each $m$. This is, for each $m$, there exists $\Omega_m \subset \Omega$ with $\p[\Omega_m]=1$ such that \eqref{eq:RAGE} holds for all $\omega\in\Omega_m$. Finally, the set $\widetilde \Omega = \bigcap_m \Omega_m$ is such that $\p[\widetilde \Omega]=1$ and such that \eqref{eq:RAGE} holds for all $m$ simultaneously for all $\omega \in \widetilde \Omega$.
\end{proof}
%%%%%%%%%%%%%%%%%%%%%%%%%%%%%%%%%%%%%%%%%%%%%%%%%%%%%%%%%%%%
%%%%%%%%%%%%%%%%%%%%%%%%%%%%%%%%%%%%%%%%%%%%%%%%%%%%%%%%%%%%

%%%%%%%%%%%%%%%%%%%%%%%%%%%%%%%%%%%%%%%%%%%%%%%%%%%%%%%%%%%%
%%%%%%%%%%%%%%%%%%%%%%%%%%%%%%%%%%%%%%%%%%%%%%%%%%%%%%%%%%%%
%%%%%%%%%%%%%%%%%%%%%%%%%%%%%%%%%%%%%%%%%%%%%%%%%%%%%%%%%%%%
%%%%%%%%%%%%%%%%%%%%%%%%%%%%%%%%%%%%%%%%%%%%%%%%%%%%%%%%%%%%

%%%%%%%%%%%%%%%%%%%%%%%%%%%%%%%%%%%%%%%%%%%%%%%%%%%%%%%%%%%%
%%%%%%%%%%%%%%%%%%%%%%%%%%%%%%%%%%%%%%%%%%%%%%%%%%%%%%%%%%%%

%%%%%%%%%%%%%%%%%%%%%%%%%%%%%%%%%%%%%%%%%%%%%%%%%%%%%%%%%%%%
%%%%%%%%%%%%%%%%%%%%%%%%%%%%%%%%%%%%%%%%%%%%%%%%%%%%%%%%%%%%

%%%%%%%%%%%%%%%%%%%%%%%%%%%%%%%%%%%%%%%%%%%%%%%%%%%%%%%%%%%%
%%%%%%%%%%%%%%%%%%%%%%%%%%%%%%%%%%%%%%%%%%%%%%%%%%%%%%%%%%%%
%%%%%%%%%%%%%%%%%%%%%%%%%%%%%%%%%%%%%%%%%%%%%%%%%%%%%%%%%%%%
%%%%%%%%%%%%%%%%%%%%%%%%%%%%%%%%%%%%%%%%%%%%%%%%%%%%%%%%%%%%
%\bibliography{schrodinger}{}

\begin{thebibliography}{99}
%%\addcontentsline{toc}{chapter}{Bibliography}
%

\bibitem{An58} P. Anderson, \textit{Absence of diffusion in certain random 
lattices}, Phys. Rev. {\bf 109}, 1492-1505 (1958).

\bibitem{AENSS} M. Aizenman, A. Elgart, S. Naboko, J. Schenker, G. Stolz, \textit{Moment analysis for localization in random Schro\"odinger operators}, Invent. Math., {\bf 163}, 343-413 (2006).

\bibitem{AM} M. Aizenman, S. Molchanov, \textit{Localisation at large disorder and at extreme energies:
an elementary derivation}, Comm. Math. Phy. {\bf 157}, 245-278 (1993).
%
\bibitem{ASW} M. Aizenman, R. Sims, S. Warzel, \textit{Stability of the absolutely continuous spectrum of random Schr\"odinger operators on tree graphs},
 Probab. Theory Related Fields {\bf 136}, 363-394 (2006).
%
%
\bibitem{AW} M. Aizenman, S. Warzel, \textit{Random Operators: Disordered effects on Quantum spectra and dynamics}, Graduate Studies in Mathematics, vol 168 AMS (2016).
%
%
\bibitem{B1} J. Bourgain, \textit{ On random Schr\"odinger operators on $\Z^2$},
 Discr. Contin. Dyn. Syst. {\bf 8}, 1-15 (2002).

\bibitem{B2} J. Bourgain, \textit{Random lattice Schr\"odinger operators with decaying potential: some higher dimensional phenomena}, 
Geometric Aspects of Functional Analysis, Lectures Notes in Math. {\bf 1807}, 70-98, Springer, Berlin-Heidelberg (2003).

\bibitem{BMT00} O. Bourget, G. Moreno Flores, A. Taarabt,
\textit{One-dimensional discrete Anderson model in a decaying random potential: from ac spectrum to dynamical localization}, preprint.
%
\bibitem{BMT01} O. Bourget, G. Moreno Flores, A. Taarabt,
\textit{One-dimensional discrete Dirac model in a decaying random potential I: spectrum and dynamics}, preprint.

\bibitem{BDFG} V. Bucaj, D. Damanik, J. Fillman, V. Gerbuz, T. VandenBoom, F. Wang, Z. Zhang,
\textit{Localization for the one-dimensional Anderson model via positivity and large deviations for the Lyapunov exponent},
Trans. Amer. Math. Soc. {\bf 372}, 3619-3667 (2019)

\bibitem{Car} R. Carmona, \textit{Exponential localization in one dimensional disordered systems}, Duke Math. J. {\bf 49} 191-213 (1982).
%
%
\bibitem{CHN} J.M. Combes, P. Hislop, S. Nakamura, \textit{The Lp-Theory of the Spectral Shift Function, the Wegner Estimate, and the Integrated Density of States for Some Random Operators}, Commun. Math. Phys. {\bf 218}, 113-130 (2001).

\bibitem{CFKS} H.L. Cycon, R.G. Froese, W. Kirsch, B. Simon, \textit{Schr\"odinger operators with applications to quantum Mechanics and global Geometry}, Texts and Monographs in Physics, Springer Study Edition, Springer-Verlag, Berlin, (1987).
%
\bibitem{CKM} R. Carmona, A. Klein, F. Martinelli, \textit{Anderson Localization for Bernoulli and Other Singular Potentials}, Commun. Math. Phys. {\bf 108}, 41-66 (1987).

\bibitem{DG} D. Damanik, A. Gorodetski,
\textit{An extension of the Kunz-Souillard approach to localization in one dimension and applications to almost-periodic Schr\"odinger operators},
Adv. Math. {\bf 297} (2015).

\bibitem{DaSiSt} D. Damanik, R. Sims, G. Stolz, \textit{Localization for one-dimensional, continuum, Bernoulli-Anderson models}, Duke J. Math, {\bf 114}, 1, 59-100 (2002)

\bibitem{DS} D. Damanik, G. Stolz, \textit{A continuum version of the Kunz–Souillard approach to localization in one dimension}, Journal f\"ur die reine und angewandte Mathematik (Crelles Journal), {\bf 660}, 99-130 (2011).
%

%
\bibitem{DeRJLS1} R. Del Rio, S. Jitomirskaya, Y. Last and B. Simon, \textit{What is localization?}, Phys. Rev. Lett. 75, 117-119 (1995).

\bibitem{DeRJLS2} R. Del Rio, S. Jitomirskaya, Y. Last and B. Simon, \textit{Operators with singular continuous spectrum IV: Hausdorff dimensions, rank one pertubations and localization}, J. Anal. Math. 69 153-200 (1996).
%

\bibitem{D} F. Delyon, \textit{Appearance of a purely singular continuous spectrum in a class of random Schr\"odinger operators},
J. Statist. Phys. {\bf 40}, 621-630 (1985).

\bibitem{DLS} F. Delyon, H. Kunz, B. Souillard, \textit{One-dimensional wave equations in disordered media},
J. Phys. A: Math. Gen. {\bf 16}, 25-42, 1983

\bibitem{DSS} F. Delyon, B. Simon, B. Souillard,
\textit{From power pure point to continuous spectrum in disordered systems},
Ann. Henri Poincaré, {\bf 42} vol. 6, 283-309 (1985).

%
%
\bibitem{FGKM} A. Figotin, F. Germinet, A. Klein, P. M\"uller, \textit{Persistence of Anderson localization in Schr\"odinger operators with decaying random potentials}, Ark. Mat. {\bf 45} 15-30 (2007).
%
\bibitem{FHS} R. Froese, D. Hasler, W. Spitzer, \textit{Absolutely continuous spectrum for the Anderson model on a tree: a geometric proof of Klein's theorem},
Comm. Math. Phys. {\bf 269}, 239-257 (2007).
%
\bibitem{GK1} F. Germinet, A. Klein, \textit{Bootstrap multiscale analysis and localization in random
media}, Commun. Math. Phys. {\bf 222}, 415-448 (2001).
%
%

\bibitem{GZ} L. Ge, X. Zhao, 
\textit{Exponential dynamical localization in expectation for one dimensional Anderson model},
to appear in J. Spect. Theory.

\bibitem{GKS} F. Germinet, A. Klein, J. Schenker, \textit{Dynamical delocalization in
random Landau Hamiltonians}, Ann. Math. {\bf 166}, 215-244 (2007).
%
\bibitem{GKT} F. Germinet, A. Kiselev, S. Tcheremchantsev,
\textit{Transfer matrices and transport for Schr\"odinger operators},
Ann. Inst. Fourier {\bf 54}, 787-830 (2004).
%
%
\bibitem{GMP} I. Goldsheid, S. Molchanov, L. Pastur, \textit{A pure point spectrum of the stochastic
one-dimensional Schr\"odinger equation}, Funct. Anal. Appl. {\bf 11}, 1-10 (1977).
%
\bibitem{GT} F. Germinet, A. Taarabt, 
\textit{Spectral  properties  of  dynamical  localization  for Schr\"odinger operators}, Rev. Math. Phys. {\bf 25}, 9 (2013).
%
\bibitem{HSS} E. Hamza, R. Sims, G. Stolz, \textit{A note on fractional moments for the one-dimensional continuum Anderson model}, J. Math. Anal. Appl. {\bf 365}, 435-446 (2010)

%\bibitem{JL} V. Jak\u{s}i\'{c}, Y, Last,
%\textit{Spectral structure of Anderson type Hamiltonians},
%Inven. Math. {\bf 141}, 561-577 (2000).
%
%
\bibitem{JZ} S. Jitomirskaya, X. Zhu,
\textit{Large deviations of the Lyapunov exponent and localization for the 1D Anderson model}, 
Comm. Phys. Math. {\bf 370}, 311-324 (2019).

\bibitem{Kl} A. Klein, \textit{Extended states in the Anderson model on the bethe lattice}, Adv. in Math. {\bf 133}, 163-184 (1998).
%
\bibitem{Kr} M. Krishna, \textit{Anderson model with decaying randomness: existence of extended states}, Proc. Indian Acad. Sci. (Math. Sci.) {\bf 100},
285-294 (1990).
%
\bibitem{KKO} W. Kirsch, M. Krishna, J. Obermeit,
\textit{Anderson model with decaying randomness: mobility edge},
Math. Z. {\bf 235}, 421-433 (2000).
%
%
\bibitem{KLS} A. Kiselev, Y. Last, B. Simon,
\textit{Modified Pr\"ufer and EFGP transforms and the spectral analysis of one-dimensional schr\"odinger operators},
Comm. Math. Phys. {\bf 194}, 1-45 (1998).
%
%\bibitem{KRS} A. Kiselev, C. Remling, B. Simon, \textit{Effective perturbation methods for one-dimensional Schr\"odinger operators}, J. Differential Equations {\bf 151}, 290-312 (1999).
%
\bibitem{KS} H. Kunz, B. Souillard, \textit{Sur le spectre des op\'erateurs aux diff\'erences finies al\'eatoires}, Comm. Math. Phys. {\bf 78}, 201-246 (1980).
%
%
\bibitem{KU} S. Kotani, N. Ushiroya, 
\textit{One-dimensional Schr\"odinger operators with random decaying potentials},
Comm. Math. Phys. {\bf 115}, 247-266 (1988).
%
\bibitem{Si82} B. Simon, 
\textit{Some Jacobi matrices with decaying potential and dense point spectrum}, Comm. Math. Phys. {\bf 87}, 253-258 (1982).

\bibitem{CBD} Stollman, P., \textit{Caught by disorder}, Progress in Mathematical Physics 20, Springer (2001). 
%\bibitem{Si95} B. Simon, \textit{Spectral Analysis of rank one perturbations and applications},
%CRM Lectures Notes, Vol. 8, Amer. Math. Soc, Providence, RI, 109-149 (1995).
%
\end{thebibliography}
%\bibliographystyle{plain}
%\end{document}

\end{document}